\documentclass[letter,11pt]{article}
\usepackage{amssymb,latexsym}
\usepackage{graphicx,amsmath,amssymb}
\usepackage[usenames]{pstricks}
\usepackage{epsfig}
\usepackage{color}
\usepackage[usenames]{xcolor}
\usepackage{amsthm}
\usepackage{mathrsfs}
\usepackage{stmaryrd}
\usepackage[pdfstartview=FitH,backref,pdfpagemode=None,colorlinks=true,citecolor=darkblue,urlcolor=black,linkcolor=bluetxt,filecolor=black]{hyperref}
\usepackage{caption}
\usepackage{fullpage}
\usepackage{phuong}

\newcommand{\NCFIPS}{non-commutative IPS}
\newcommand{\TT}{\mathbf T}
\newcommand{\LL}{\Lambda}
\newcommand{\FF}{\overline F}

\newtheorem{theorem}{Theorem}[section]
\newtheorem{lemma}[theorem]{Lemma}
\newtheorem*{lemma*}{Lemma}

\newtheorem{definition}[theorem]{Definition}

\newtheorem{corollary}[theorem]{Corollary}

\newtheorem*{remark*}{Remark}

\newtheorem*{notation*}{Notation}

\newtheorem*{claim*}{Claim}

\newtheorem*{main-open}{Main Open Problem}

\newtheorem*{theorem*}{Theorem}

\newenvironment{note}{\QuadSpace\par\noindent{\bf Note}:}{\HalfSpace}

\newenvironment{proofclaim}{\QuadSpace\par\noindent{\textit {Proof of claim}:}}
{$\Box_{\textrm{\,claim}}$ \HalfSpace}

\newcommand\F{\ensuremath{\mathbb F}}
\newcommand\N{\ensuremath{\mathbb N}}
\newcommand\Z{\ensuremath{\mathbb Z}}
\newcommand{\ignore}[1]{}

\newcommand{\poly}{\hbox{poly}}
\newcommand{\cd}{\cdot}
\renewcommand{\l}{\ell}
\newcommand{\nha}{\ensuremath{\nonHomoABP}} 

\newcommand{\Base}{\mbox{}\\ \ind{\textit{Base case: }}}
\newcommand{\Induction}{\mbox{}\\ \ind{\textit{Induction step: }}}
\newcommand{\case}[1]{\ind\textbf{Case #1}:\,}
\newcommand{\induction}{\Induction}

\newcommand{\synEqual}{\ensuremath{\vdash^*}}

\newcommand {\ind} {\noindent}

\newcommand {\para}[1] {\paragraph{#1}}

\DeclareMathAlphabet{\mathitbf}{OML}{cmm}{b}{it}

\newcommand{\freea}{\ensuremath{\F\langle X\rangle}}

\newcommand{\QuadSpace}{\vspace{0.25\baselineskip}}
\newcommand{\HalfSpace}{\vspace{0.5\baselineskip}}

\newcommand{\EndProof}{ \hfill \vrule width 1ex height 1ex depth 0pt }

\definecolor{bluetxt}{rgb}{0,0,.6}
\definecolor{myred}{rgb}{0.6,0.0,0.1}
\definecolor{greentxt}{rgb}{0,.5,0}
\definecolor{redtxt}{rgb}{0.1,0.1,0.65}
\definecolor{purpletxt}{rgb}{0.6,0.1,0.7}
\definecolor{black}{rgb}{.0,.0,.0}
\definecolor{verydarkblue}{rgb}{.0,.0,.4}
\definecolor{darkblue}{rgb}{.0,.0,.53}
\definecolor{lightgray}{rgb}{.7,.7,.7}

\newlength{\defbaselineskip}
\newcommand{\doublespacing}{\setlength{\baselineskip}{1.0\defbaselineskip}}
\renewcommand{\b}{\beta}

\newcommand{\set}[1]{\left\{#1\right\}}

\newcommand{\nx}[1]{#1_1,\ldots,#1_{n}}

\renewcommand{\a}{\alpha}
\renewcommand{\b}{\beta}
\newcommand{\commF}{\phi_{L_i,L_j}}
\newcommand{\proofbound}[1]{\left(\sum_{\l\in #1} |L_\l|\right)^4}

\newcommand{\boolean}{Boolean}

\newcommand{\g}{g}
\renewcommand{\.}{ ,\ldots,}
\newcommand{\bool}{bool}
\newcommand{\FPC}{\ensuremath{\mathcal{F\mbox{\rm-}PC}}}
\newcommand{\xQx}{\overline{x},\overline{Q}^\phi(\overline{x})}
\newcommand{\xZero}{\overline{x},\overline{0}}
\newcommand{\pQx}{
        \ensuremath{
                \overline{p},
                \overline {Q}_{\bool}^{\phi}(\overline p)
        }
}
\newcommand{\pZero}{\overline{p},\overline{0}}

\newcommand{\lequal}{\leftrightarrow}
\newcommand{\Fb}{F_{\bool}}
\newcommand{\Fz}{\neg\Fb \left(\pZero\right)}
\newcommand{\Fq}{\Fb \left(\pQx\right)}
\newcommand{\truthTableki}{\normalfont\textsf{Truth}([\kappa_i],\overline{p})}
\newcommand{\truthTablekj}{\normalfont\textsf{Truth}([\kappa_j],\overline{p})}
\newcommand{\truthTablek}{\normalfont\textsf{Truth}([\kappa],\overline{p})}

\newcommand{\x}{D^u}
\newcommand{\D}[2]{\mathcal{D}_{#1}^{#2}}
\newcommand{\Du}[1]{\D{#1}{u}}
\newcommand{\Dv}[1]{\D{#1}{v}}
\newcommand{\Dw}[1]{\D{#1}{w}}
\newcommand{\pd}[1]{#1_{pd}}
\newcommand{\pdu}{\pd{u}}

\newcommand{\Gd}{N_{\pdu}}
\newcommand{\hhat}{\ensuremath{\widehat{\homoji}}}

\newcommand{\vsink}{\ensuremath{v_{\rm sink}}}

\newcommand{\h}[1]{^{(#1)}}

\newcommand{\hF}[1]{{F}^{(#1)}}
\newcommand{\hFbool}[1]{{F}_{bool}^{(#1)}}

\newcommand{\bF}{F^\star}
\newcommand{\vin}{\ensuremath{v_{\rm source}}}

\newcommand{\homojBool}{F_{bool}(\overline p)}
\newcommand{\zeroMatrix}{\Lambda_i}

\newcommand{\zeroMatrixFinal}{\Lambda_{d}}

\newcommand{\homoji}{\ensuremath \overline F_{i}}
\newcommand{\homojiFinal}{\overline F_{d}}

\newcommand{\homojiBool}{\overline \homojBool_i}
\newcommand{\homojiBoolt}{\homojBool_{i,t}}
\newcommand{\homojiBooltFinal}{\homojBool_{d,t}}
\newcommand{\homojiBoolFormert}{\homojBool_{i-1,t}}

\newcommand{\nodeF}{F^{\bullet}}
\newcommand{\trn}{{\rm{tr}}}
\newcommand{\pol}{\ensuremath{\widehat}}

\newcommand{\truthTable}{{\normalfont \textsf{Truth}}([\phi],\overline{p})}

\newcommand{\Qi}{Q_{i_{\\ \bool}}^{\phi}(\overline{p})}
\newcommand{\Qbar}{Q_{\bool}^{\phi}}

\newcommand{\sequiv}{\leftrightarrow}

\author
{
        \normalsize
       {Fu Li}\thanks{Institute for Interdisciplinary Information Sciences. 
Supported in part by NSFC grant 61373002.
                                \texttt{fuli.theory.research@gmail.com}
                  }
           \vspace{11pt}
           \\
       \begin{minipage}{100pt}
           \center \small
           \textit{Tsinghua University}
       \end{minipage}
 \and
        \normalsize
       {Iddo Tzameret}\thanks{
                              Department of Computer Science.
                              Supported in part by NSFC grant 61373002.
                              \texttt{Iddo.Tzameret@rhul.ac.uk}
                              }
                       \vspace{11pt} 
                       \\
       \begin{minipage}{110pt}
       \center \small
       \textit{Royal Holloway, University of London}
       \end{minipage}
  \and \normalsize
       {Zhengyu Wang}\thanks{
                        Department of Computer Science. 
                        \texttt{wangsincos@163.com}
                       }
                       \vspace{11pt} 
                       \\
       \begin{minipage}{100pt}
       \center \small
       \textit{Harvard University}
       \end{minipage}
}

\date{}

\usepackage[normalem]{ulem}
\usepackage{pifont}

\usepackage{bbding}
\usepackage{wasysym}

\begin{document}

\title{Characterizing Propositional Proofs \\ as Non-Commutative Formulas\footnote{An extended abstract of this work entitled ``Non-commutative Formulas and Frege Lower Bounds: a New Characterization of Propositional Proofs'' appeared in \emph{Proceedings of the 30th Annual Computational Complexity Conference} (CCC): June 17-19, 2015.}
}
\maketitle
\vspace{-15pt} 
\begin{abstract}
Does every Boolean tautology have a short propositional-calculus proof? Here,  a propositional-calculus (i.e., Frege) proof  is any  proof starting from a set of axioms and deriving new Boolean formulas using a fixed set of sound derivation rules. Establishing  any super-polynomial size lower bound on Frege proofs (in terms of the size of the formula proved) is a major open problem in proof complexity, and among a handful of fundamental hardness questions in complexity theory by and large. Non-commutative arithmetic formulas, on the other hand, constitute a quite weak computational model, for which exponential-size lower bounds were shown already back in 1991 by Nisan [STOC 1991], using  a particularly transparent argument.

In this work we show that Frege lower bounds  in fact follow from  corresponding size lower bounds on non-commutative formulas computing certain polynomials (and that such lower bounds on non-commutative formulas must exist, unless \NP=\coNP). More precisely, we demonstrate a natural association between tautologies $T$ to non-commutative polynomials $p$, such that:
\begin{itemize}
\item[\Pisymbol{pzd}{86}] if  $T$ has a polynomial-size Frege proof then  $p$ has a polynomial-size non-commutative arithmetic formula;
and conversely, when $T$ is a DNF, if  $p$ has a polynomial-size non-commutative arithmetic formula over $GF(2)$ then  $T$ has a Frege proof of quasi-polynomial size.
\end{itemize}
\ind The argument is a characterization of Frege proofs as non-commutative formulas: we show that the Frege  system is (quasi-) polynomially equivalent to a \textit{non-commutative Ideal Proof System }(IPS), following  the recent work of Grochow and Pitassi [FOCS 2014] that introduced a propositional proof system in which proofs are arithmetic circuits, and the work in \cite{Tza11-I&C} that  considered adding the commutator as an axiom in  algebraic propositional proof systems.
This also gives a characterization of propositional Frege proofs in terms of  (non-commutative) arithmetic formulas that is tighter than (the formula version of IPS) in Grochow and Pitassi [FOCS 2014].
%
%
%
\end{abstract}


%
\doublespacing
\section{Introduction}

\subsection{Propositional proof complexity} The field of propositional proof complexity aims to understand and analyze the computational resources required to prove  propositional statements. The problems the field poses  are  fundamental, difficult, and of central
importance to computer science and complexity theory as demonstrated by the seminal work of Cook and Reckhow \cite{CR79}, who showed the immediate relevance of these problems to the \NP~vs.~\coNP\ problem (and thus to the \Ptime~vs.~\NP\ problem).

Among the major unsolved questions in propositional proof complexity, is whether the standard propositional logic calculus, either in the form of the Sequent Calculus, or equivalently, in the axiomatic form of Hilbert style proofs (i.e.,  Frege proofs), is polynomially bounded; that is, whether every propositional tautology---namely,
a formula that is satisfied by every assignment---has a proof whose size is polynomially bounded in the size of the formula proved (alternatively and equivalently, we can think of unsatisfiable formulas and their refutations). Here, we consider the size of proofs as the number of symbols it takes to write them down, where each formula in the proof is written as a Boolean \textit{formula}  (in other words we count the total number of logical gates appearing in the proof). 

It is known since Reckhow work \cite{Rec76:PhD} that all Frege
proof-systems\footnote{Formally, a Frege proof system is any propositional proof system with a fixed number of axiom schemes and sound derivation rules that  is also implicationally complete, and in which proof-lines are written as propositional formulas (see Definition \ref{def:Frege_system}).}~(as well as the Gentzen sequent calculus with the cut rule \cite{Gen35}) are polynomially equivalent to each other, and hence it does not matter precisely which rules, axioms, and logical-connectives we use in the system.
Nevertheless, for concreteness, the reader can think of the Frege proof system as the following simple one (known as \textit{Schoenfield's system}),  consisting of only three
axiom schemes (where $A\to B$ is an abbreviation of $\neg A\lor B$; and $A,B,C$ are
any propositional formulas): \vspace{-5pt} 
\begin{gather*}
A\to(B\to A)
\\
(\neg A \to \neg B)\to ((\neg A\to B)\to A)
\\
(A\to (B\to C))\to ((A\to B)\to (A\to C)),
\end{gather*}
and a single inference rule (known as \textit{modus ponens}):
$$
\hbox{from $A$ and $A\to B$, infer $B$}\,.
$$

Complexity-wise,  Frege is considered a very strong proof system alas a poorly understood one. The qualification \textit{strong} here has several meanings: first, that no super-polynomial lower bound is known for Frege proofs. Second, that there are not even good hard candidates for the Frege system (see \cite{BBP95,Razb15-annals,Kra:book11,LT13} for further discussions on hard proof complexity candidates). Third, that for most hard instances (e.g., the pigeonhole principle and
Tseitin tautologies) that are known to be hard  for weaker systems (e.g., resolution, cutting planes, etc.), there \textit{are }known polynomial bounds on Frege proofs. Fourth, that proving super-polynomial lower bounds on Frege proofs seems to a certain extent out of reach of current techniques (and believed by some to be even \emph{harder} than proving explicit circuit lower bounds \cite{Razb15-annals}). And finally, that by the common (mainly informal) correspondence between circuits and proofs---namely, the correspondence between a circuit-class $\cal C$ and a proof system in which every proof-line is written as a circuit\footnote{To be more precise, one has to associate a circuit class $\cal C$ with a proof system in which a \emph{family} of proofs is written such that every proof-line in the family is a circuit family from $\cal C$.} from $\cal C$---Frege system corresponds to the circuit class of polynomial-size $\log(n)$-depth circuits denoted \NCOne\ (equivalently, of polynomial-size formulas \cite{Spi71}), considered to be a strong computational model for which no (explicit) super-polynomial lower bounds are currently known.

Accordingly, proving lower bounds on Frege proofs is considered an extremely hard task. In fact, the best lower bound known today is only  quadratic, which uses a fairly simple syntactic argument  \cite{Kra95}. If we put further impeding restrictions on Frege proofs, like restricting the depth of each formula appearing in a proof to a certain fixed constant, exponential lower bounds can be obtained \cite{Ajt88,PBI93,PBI93}. Although these constant-depth Frege exponential-size lower bounds  go back to Ajtai's result from 1988, they are still in some sense the  state-of-the-art in proof complexity lower bounds (beyond the important developments on weaker proof systems, such as resolution and its comparatively weak extensions). Constant-depth Frege  lower bounds  use  quite involved probabilistic arguments, mainly  specialized switching lemmas tailored for specific tautologies (namely, counting tautologies, most notable of which are the Pigeonhole Principle tautologies). Even random $k$CNF formulas near the satisfiability threshold are not known to be hard for constant-depth Frege (let  alone hard for
[unrestricted depth] Frege).

All of the above goes to emphasize the importance, basic nature and difficulty in  understanding the complexity of strong propositional proof systems, while showing how little is actually known about these systems. \smallskip

\subsection{Prominent directions for understanding propositional proofs} As we already mentioned, there is a guiding line in proof complexity which states a  correspondence between the complexity of circuits and the complexity of proofs. This correspondence is mainly informal, but there are seemingly good indications showing it might be more than a superficial analogy. One of the most compelling evidence for this correspondence is that there is a formal correspondence (cf.~\cite{CN10} for a clean formulation
of this) between the first-order logical theories of bounded arithmetic (whose axioms state the existence of sets taken from a given complexity class $\mathcal C$) to propositional proof systems (in which proof-lines are circuits from $\mathcal C$).

Another aspect of the informal correspondence between circuit complexity and proof complexity is that circuit hardness sometimes can be used to obtain proof complexity hardness.
The most notable example of this are the lower bounds on constant-depth Frege proofs mentioned above: constant-depth Frege proofs can be viewed as propositional calculus operating with \ACZ\ circuits, and the known lower bounds on constant depth Frege proofs (cf. \cite{Ajt88,KPW95,PBI93}) use techniques borrowed from \ACZ\ circuits lower bounds. The success in moving from circuit hardness towards proof-complexity hardness has spurred a flow of attempts to obtain lower bounds on proof systems other than constant depth Frege. For example, Pudl{\'a}k \cite{Pud99} and Atserias et al.~\cite{AGP01} studied proofs based on monotone circuits, motivated by known exponential lower bounds on monotone circuits \cite{Razb85}. Raz and Tzameret  \cite{RT06,RT07,Tza08:PhD} investigated algebraic proof systems operating with multilinear formulas, motivated by lower bounds on multilinear formulas for the determinant, permanent and other explicit polynomials \cite{Raz04a,Raz04b}. Atserias et al.~\cite{AKV04}, Kraj\'{i}\v{c}ek \cite{Kra07} and Segerlind \cite{Seg07} have considered proofs operating with ordered binary decision diagrams (OBDDs), and the second author \cite{Tza11-I&C} initiated the study of proofs operating with non-commutative formulas (see Sec.~\ref{sec:comparison} for a comparison with the current work).\footnote{We
do not discuss here the important thread of results whose aim is to establish conditional
lower bounds based on Nisan-Wigderson generators. This direction
was developed in e.g.~\cite{ABSRW00,Razb15-annals,Kra04,Kra10-forcing}.}

Until quite recently it was unknown  whether the correspondence between proofs and circuits is two-sided, namely, whether proof complexity hardness (of concrete known proof systems) can imply any computational hardness.
An initial example of such an implication from proof hardness to circuit hardness was given by Raz and Tzameret \cite{RT06}. They showed that a separation between algebraic proof systems operating with arithmetic  circuits and multilinear arithmetic circuits, resp., for an explicit family of polynomials, implies a separation between arithmetic circuits and multilinear arithmetic circuits. 

In a recent significant development about the complexity of strong proof systems, Grochow and Pitassi \cite{GP14} demonstrated a much stronger correspondence. They introduced a natural propositional proof system, called the \textit{Ideal Proof System }(IPS for short), for which \textit{any }super-polynomial size lower bound on IPS  implies a corresponding size lower bound on arithmetic circuits, and formally, that the permanent does not have polynomial-size  arithmetic circuits.   The IPS is defined as follows:

\begin{definition}[Ideal Proof System (IPS) \cite{GP14}]\label{def:orig-IPS} Let $F_1(\overline x),\ldots,F_m(\overline x)$ be a system of polynomials in the variables $x_1,\ldots,x_n $, where the polynomials $x_i^2-x_i$, for all $1\le i\le n$, are part of this system. An \emph{IPS refutation (or certificate)} that the $F_i$'s polynomials have no common 0-1 solutions is a polynomial $C(\overline x,\overline y)$ in the variables $x_1,\ldots,x_n $ and $y_1,\ldots,y_{m} $, such that:\vspace{-5pt}

\begin{enumerate}

\item $F(\nx{x},\overline 0) = 0$; and \vspace{-8pt}

\item  $F(\nx{x}, F_1(\overline x),\ldots, F_m(\overline x)) = 1.$
\end{enumerate}

\end{definition}

The essence of IPS is that a proof (or refutation) is a \textit{single }polynomial that can be written simply as an arithmetic \textit{circuit} or \emph{formula}. The advantage of this formulation is that now we can obtain direct connections between circuit/formula  hardness (i.e., ``computational hardness'') and hardness of proofs. Grochow and Pitassi showed indeed that a lower bound on IPS written as an arithmetic circuit implies that the permanent does not have polynomial-size algebraic circuits (Valiant's conjectured separation $\sf{VNP}\neq\sf{VP}$ \cite{Val79:ComplClass}); And similarly,  a lower bound on IPS written as an arithmetic \textit{formula }implies that the permanent
does not have polynomial-size algebraic formulas  ($\sf VNP\neq\sf VP_e$, ibid).

Under certain assumptions, Grochow and Pitassi \cite{GP14} were able to connect  their result to  standard propositional-calculus proof systems, i.e., Frege and Extended Frege. Their assumption was the following: \textit{Frege has polynomial-size proofs of the statement expressing that the PIT for arithmetic formulas is decidable by polynomial-size Boolean circuits }(\textit{PIT for arithmetic formulas} is the problem of deciding whether an input arithmetic formula computes the [formal] zero polynomial).
They showed that\footnote{We focus only on the relevant results about Frege proofs from \cite{GP14} (and not the results about Extended Frege in \cite{GP14}; the latter proof system operates, essentially,  with Boolean circuits, in the same way that Frege operates with Boolean formulas (equivalently \NCOne\ circuits)).}, under this assumption super-polynomial lower bounds on Frege proofs imply
that the permanent does not have polynomial-size arithmetic circuits. This, in turn, can be considered as a (conditional) justification for the apparent long-standing difficulty of proving lower bounds on strong proof systems.

\subsection{Overview of  results and proofs}

\subsubsection{Sketch}
In this work we give a novel characterization of the propositional calculus---a fundamental
and prominent object by itself---and by this contribute to the understanding of strong propositional proof systems, and to the fundamental search for lower bounds on these proofs. We formulate a very natural proof system, namely a non-commutative variant of the ideal proof system, which we show captures \textit{unconditionally} (up to a quasi-polynomial-size increase, and in some cases only a polynomial increase\footnote{We establish a slightly stronger characterization: the non-commutative IPS polynomially
simulates Frege; and conversely, the complexity in which Frege simulates  the non-commutative IPS depends on the degree of the non-commutative IPS refutation; e.g., the simulation
is \textit{polynomial} when refutations are of logarithmic degrees (see note after Theorem \ref{thm:intro:Frege_sim_ncIPS}).}%
) propositional Frege proofs.  A proof in the  non-commutative IPS is simply a \textit{single non-commutative \textit{polynomial}} written as a non-commutative formula. 

Our results thus give a compelling and simple new characterization of the proof complexity of propositional Frege proofs and brings new hope for achieving lower bounds on strong proof systems, by reducing the task of lower bounding Frege proofs to the following seemingly much more manageable task: proving matrix rank lower bounds on the matrices associated with certain non-commutative polynomials (in the sense of Nisan \cite{Nis91}; see below for details).

The new characterization also tightens the recent results of Grochow and Pitassi \cite{GP14}
in the following sense:
%
\vspace{-5pt} 
\begin{enumerate}
\item[(i)]
The non-commutative IPS is \textit{polynomial-time checkable}---whereas the original IPS was checkable in probabilistic polynomial-time; and
\vspace{-6pt}
\item [(ii)] Frege proofs \textit{unconditionally} quasi-polynomially simulate the non-commutative IPS---whereas Frege  was shown to  efficiently simulate IPS only assuming that the decidability of PIT for (commutative) arithmetic formulas   by  polynomial-size circuits is efficiently provable in Frege.
\end{enumerate}
The tighter result shows that, at least for Frege, and in the framework of the ideal proof system, lower bounds on Frege proofs do not  necessarily entail in themselves very strong computational lower bounds.

\subsubsection{Some preliminaries: non-commutative polynomials and formulas}

A \textit{non-commutative polynomial} over a given  field \F\ and with the variables $X:=\{x_1,x_2,\ldots\}$ is a formal sum of monomials with coefficients from \F\ such that the product of variables is non-commuting.  For example, $x_1 x_2-x_2 x_1+x_3 x_2 x_3^2-x_2 x_3^3,\  \, x_1 x_2-x_2 x_1$ and $0$ are three distinct polynomials in $\freea$.
The ring of non-commutative polynomials with variables $X$ and coefficients from \F\ is denoted \freea. 

A \textit{polynomial} (i.e., a \emph{commutative}
polynomial) over a field is defined in the same way as a non-commutative
polynomial  except that
 the product of variables \textit{is} commutative; in other words, it
is a sum of (commutative) monomials.

A \emph{non-commutative arithmetic formula} (non-commutative formula for short) is a fan-in two labeled tree, with edges directed from leaves towards the root, such that the leaves are labeled with field elements (for a given field \F)  or variables $x_1,\ldots,x_n, $         and internal nodes (including the root) are labeled with a plus $+$ or product  $\times$ gates. A product gate has an order on its two children (holding the order of non-commutative product). A non-commutative formula computes a non-commutative polynomial in the natural way (see  Definition \ref{def:nonc_formula}).

Exponential-size lower bounds on non-commutative formulas (over any field) were established by Nisan \cite{Nis91}. The idea (in retrospect) is quite simple: first transform a non-commutative formula into an algebraic branching program (ABP; Definition \ref{def:ABP}); and then show that the number of nodes in the $i$th layer of an ABP computing a degree $d$ homogenous non-commutative polynomial $f$ is bounded from below by the rank of the degree $i$-partial-derivative matrix of $f$.%
\footnote{The degree $i$ partial derivative matrix of $f$ is the matrix whose ro`ws are all non-commutative monomials of degree $i$ and columns are all non-commutative monomials of degree $d-i$, such that the entry in row $ M$ and column $N$ is the coefficient of the $d$ degree monomial $M\cd N$ in $f$.} Thus, lower bounds on non-commutative formulas follow from quite immediate rank arguments (e.g., the partial derivative matrices associated with the permanent and determinant can easily be shown to have high ranks).


\subsubsection{Non-commutative ideal proof system}\label{sec:intro:ncIPS}

Recall the IPS refutation system from  Definition \ref{def:orig-IPS} above.  We use the idea introduced in \cite{Tza11-I&C}, which considered adding the commutator $x_1x_2-x_2x_1$ as an axiom in propositional algebraic proof systems, to define a refutation  system that polynomially simulates Frege:\smallskip

\begin{definition}[Non-commutative IPS]\label{def:intro:non-commutative-IPS}
Let $\F$ be a field. Assume that $F_1(\overline x) = F_2(\overline x) = \cdots = F_m(\overline x) = 0$ is a system of non-commutative polynomial equations from $\F\langle x_1,\ldots,x_n\rangle$, and suppose that the following set of  equations  (axioms) are included in the $F_i(\overline x)$'s:
\begin{description}
\item[\quad Boolean axioms:]

\ $ x_i\cd (1-x_i)\,, $ \, \, for all $\,1\le i\le n\,$;\vspace{-8pt}
\item[\quad Commutator axioms:] \ $x_i\cd x_j -x_j\cd x_i$\,, \,\,  for all $\,1\le i< j\le n\,.$
\end{description}
Suppose that the $F_i(\overline x)$'s have no common $0$-$1$ solutions.\footnote{One can check that the $F_i(\overline x)$'s have no common $0$-$1$ solutions in \F\ iff they do not have a common 0-1 solution in every \F-algebra.}   A \textbf{non-commutative IPS refutation} (or \emph{certificate}) that the system of $F_i(\overline x)$'s  is unsatisfiable is a non-commutative polynomial $\mathfrak{F}(\overline x, \overline y)$ in the variables
$\nx{x}$ and $y_1,\ldots,  y_m$ (i.e. $\mathfrak{F}\in\F\langle \overline x, \overline y\rangle)$, such that:
\begin{enumerate}

\item \label{it:1 in ncIPS def}
 $\mathfrak{F}(\nx{x},\overline 0) = 0$; and \vspace{-8pt}

\item  $\mathfrak{F}(\nx{x}, F_1(\overline x),\ldots, F_m(\overline x)) = 1.$ \label{it:2 in ncIPS def}
\end{enumerate}
We always assume that the non-commutative IPS refutation is written as a \textit{non-commutative formula}. Hence the \textit{\textbf{size}} of a non-commutative IPS refutation is the minimal size of a non-commutative formula computing the non-commutative IPS refutation.
\end{definition}

\begin{note} (i) It is important to note that identities 1 and 2 in Definition \ref{def:intro:non-commutative-IPS} are \emph{formal} identities between non-commutative polynomials. It is possible
to show that without the commutator axioms the system becomes \emph{incomplete} in the sense that there will be unsatisfiable systems of non-commutative polynomials $F_1(\overline x) = F_2(\overline x) = \cdots = F_m(\overline x) = 0$ (where the $F_i$'s include the Boolean and commutator axioms) for which there are \emph{no} non-commutative IPS refutations. 

(ii) In order to prove that a system of \emph{commutative} polynomial equations $\{P_i=0\}$ (where each $P_i$ is expressed as an arithmetic formula) has no common roots in non-commutative IPS, we write each $P_i$ \emph{as a non-commutative formula} (in some way; note that there is no unique way to do this).

\end{note}

The main result of this paper is that  the non-commutative IPS (over either $\mathbb{Q}$ or $\mathbb{Z}_q$, for any prime $q$) polynomially simulates Frege; and conversely, Frege quasi-polynomially simulates the non-commutative IPS (over $\Z_2$). We explain the results in what follows.

\subsubsection*{Non-commutative IPS simulates Frege}
For the purpose of the next theorem we use a standard translation of propositional formulas $T$  into non-commutative arithmetic formulas:
\begin{definition}[$\trn(f)$]\label{def:trn_1}
Let $\trn(x_i):=x_i$, for variables $x_i$; $\trn(\textsf{false}):=1$;  $\trn(\textsf{true}):=0$; and by induction on the size of the propositional formula: $\trn(\neg T) := 1-\trn(T)$; $\trn(T_1\lor T_2)= \trn(T_1)\cd\trn(T_2)$ and finally  $\trn(T_1\land T_2) = 1-\left(1-\trn(T_1)\right)
\cd\left(1-\trn(T_2)\right)$. 
\end{definition}
%
%

For a non-commutative \emph{formula} $f$ denote by $\pol f$ the non-commutative
\emph{polynomial} computed by $f$. Thus, $T$ is a propositional tautology  iff $\pol{
\trn(T)}=0$ for every 0-1 assignment to the  variables of the non-commutative polynomial.

\begin{theorem}[First main theorem]\label{thm:intro:ncIPS_sim_Frege}
Let \F\ be either the rational numbers $\mathbb{Q}$ or $\mathbb{Z}_q$, for a prime $q$. The non-commutative IPS refutation system, when refutations are written as non-commutative formulas over \F, polynomially simulates the Frege system. More precisely, for
every propositional tautology T, if T has a polynomial-size Frege proof then there is a non-commutative
IPS certificate (over \F) of $\hbox{ \rm \trn}(\neg T)$ that has a polynomial non-commutative formula size.
\end{theorem}

The fact that an arithmetic formula (or circuit) in the form of the IPS can simulate
a propositional Frege proof was shown in \cite{GP14}. The \textit{non-commutative }IPS,
on the other hand, 
is much more restrictive than the original (commutative) IPS: instead of using commutative
polynomials (written as arithmetic formulas) we now use non-commutative polynomials (written as non-commutative arithmetic formulas).
And as mentioned above, in order to maintain the completeness of the non-commutative IPS we must add the commutator axioms $x_ix_j-x_jx_i$ to the system. Thus, the question arises: how can we still polynomially simulate Frege in this restrictive framework?
The answer to this, which also constitutes  one of the main observation of the simulation, is that \emph{the commutator axioms are already used implicitly in propositional Frege proofs}: every
classical propositional calculus system has some (possibly implicit) structural rules that enable one to commute AND's and OR's (e.g., $A\land B$ is not the same formula
as $B\land A$, from the perspective of the propositional calculus). In other words, Frege proofs operate with formulas as purely syntactic terms, and thus commutativity of AND and OR are not free for Frege proofs. 
\smallskip 

We now sketch in more detail the proof of Theorem \ref{thm:intro:ncIPS_sim_Frege}. To simulate Frege proofs we use an intermediate proof system $\mathcal{F\mbox{\rm-}PC}$  (standing
for ``formula
polynomial  calculus'') introduced by Grigoriev and Hirsch \cite{GH03}. The   $\mathcal{F\mbox{\rm-}PC}$ proof system (Definition \ref{def:F-PC}) can be thought of as a simple variant of the
well-studied \textit{polynomial calculus} (PC) system in which polynomials are written as arithmetic formulas
(instead of sums of monomials as in PC). 

Recall that a PC-refutation, as introduced
by Clegg, Edmonds and Impagliazzo \cite{CEI96}, is simply a sequence of polynomials \textit{written as sum of monomials}, where each polynomial is either taken from the
initial unsatisfiable set of polynomials or was derived using  two algebraic  rules:
from a pair of previously derived polynomials $f$ and $g$, derive $af+bg$ (for $a,b$
field elements); and from a previously derived $f$, derive $x_i\cd f$, for any variable $x_i$. The \FPC\ proof system makes the following two changes to PC (turning it into a provably much
stronger system):  
\begin{enumerate}
\item[(i)] every polynomial in an $\mathcal{F\mbox{\rm-}PC}$-proof
is written as an \textit{arithmetic formula} (instead of as a sum of monomials)
and is treated as a purely syntactic object (like in Frege); and
\vspace{-5pt}
 
\item [(ii)] we can derive new polynomials either by the two aforementioned PC rules, \emph{or by local rewriting rules} operating on any \textit{subformula}
and expressing simple operations on polynomials (such as commutativity of addition and product, associativity, distributivity, etc.).
\end{enumerate}

Grigoriev and Hirsch \cite{GH03} showed that \FPC\ polynomially simulates Frege proofs, and that for tree-like Frege proofs the polynomial simulation yields tree-like \FPC\ proofs. Since tree-like Frege is polynomially equivalent to Frege---because Frege proofs can always be balanced to a depth that  is logarithmic in their size (cf.~\cite{Kra95} for a proof)---we get that tree-like \FPC\ polynomially simulates (dag-like) Frege proofs.
\medskip 

Therefore, to conclude Theorem \ref{thm:intro:ncIPS_sim_Frege} it suffices to prove that the non-commutative IPS polynomially simulates tree-like \FPC\ proofs. To do this, loosely speaking, we construct the non-commutative  formula tree  according to the structure
of the tree-like \FPC\ proof, line by line.  
%
\smallskip

Now, since we write refutations as non-commutative formulas we can use the polynomial-time deterministic Polynomial         Identity Testing (PIT) algorithm  for non-commutative formulas, devised by Raz and Shpilka \cite{RS04}, to check in \textit{deterministic} polynomial-time the correctness of non-commutative IPS refutations:
\begin{corollary}\label{cor:intro:Cook-Reckhow}
The non-commutative IPS is a sound and complete refutation system in the sense of
Cook-Reckhow \cite{CR79}. That is, it is a sound and complete refutation system for unsatisfiable propositional formulas in which refutations \emph{can be checked for correctness in deterministic polynomial-time}.
\end{corollary}

This should be contrasted with the original (commutative) IPS of \cite{GP14}, for which verification of refutations is done in \textit{probabilistic }polynomial time using the standard Schwartz-Zippel \cite{Sch80,Zip79} PIT algorithm.
\bigskip 

The major  consequence of Theorem \ref{thm:intro:ncIPS_sim_Frege} is that to prove a super-polynomial Frege lower bound it suffices to prove a super-polynomial lower bound on non-commutative formulas computing certain polynomials. Specifically, it is enough to prove that any non-commutative IPS certificate $\frak F(\overline x,\overline y)$ (which is simply a non-commutative polynomial) has a super-polynomial non-commutative formula size; and yet in other words, it suffices to show that any such $\frak F$ must have a super-polynomial total rank according to the associated partial-derivatives matrices in the sense of Nisan \cite{Nis91} as discussed before.

\subsubsection*{Frege simulates non-commutative IPS}
We shall prove that Frege simulates the non-commutative IPS for CNFs (this is the case considered in \cite{GP14}), over $GF(2),$ and with only a quasi-polynomial increase in size (and for some specific cases the simulation can become  polynomial). 

It will be convenient to use a translation of clauses to non-commutative formulas
which is slightly different than  Definition \ref{def:trn_1}:



\begin{definition}[$\trn'(f)$ and $Q^\phi_i$]\label{def:trn_2}
Given a Boolean formula $f$ we define its non-commutative formula translation $\trn'(f)$ as follows. Let~ $\trn'(x):= 1-x $ ~and~ $\trn'(\neg x):= x$, for $x$ a variable.
Let $\trn'(\textsf{false}):=0$;  $\trn'(\textsf{true}):=1$; and $\trn'(f_1\lor\ldots \lor f_r):=\trn'(f_1)\cdots \trn'(f_r)$
(where the sequence of products stands for a (balanced) fan-in two tree of product gates with $\trn'(f_i)$ on the  leaves).
Further, for a CNF $\phi=\kappa_1\wedge \ldots \wedge \kappa_m$, denote by $Q_i^{\phi}$ the non-commutative formula translation $\trn'(\kappa_i)$ of  the clause $\kappa_i$.
\end{definition}
%

Note that this way, the system of equations $Q_1^{\phi}=0,\ldots, Q_m^{\phi}=0$ is unsatisfiable iff $\phi$
is unsatisfiable.

%

\begin{theorem}[Second main theorem]
\label{thm:intro:Frege_sim_ncIPS}
Let $\phi=\kappa_1\wedge \ldots \wedge \kappa_m$ be a CNF  and let  $Q_1^{\phi},\ldots, Q_m^{\phi}$  denote the corresponding non-commutative formulas for the clauses of $\phi$.
If there is a \NCFIPS\ refutation of size $s$ of    $Q_1^{\phi},\ldots, Q_m^{\phi}$ over $GF(2)$,  then there is a Frege proof of size $s^{O(\log s)}$ of the tautology
$\neg\phi$. 
\end{theorem}

\begin{note} The proof of Theorem \ref{thm:intro:Frege_sim_ncIPS}
achieves in fact a slightly stronger simulation than stated. That is, our simulation shows that if the degree of the non-commutative IPS refutation is $r$ and its formula depth is $d$, then there is a Frege proof of $\neg\phi$ with size ${\rm poly}\left({d+r+1 \choose r} \cd s\right)$. And in particular, Frege \emph{polynomially} simulates non-commutative IPS refutations of $O(\log n)$ degrees (for $n$ the number of variables in the CNF). 
However, for simplicity we shall always assume that the depth $d$ of 
the non-commutative IPS formula is logarithmic in its  size (Lemma \ref{lem:balance-non-commutative-formula}
shows that we can always balance non-commutative formulas), and so explicitly we only deal with the case where $d=O(\log s)$ and $r=O(s)$. 
\end{note}


 The proof of Theorem \ref{thm:intro:Frege_sim_ncIPS} consists of several
separate steps of independent interest. From the logical point
of view, the argument is a  short Frege proof of a \textit{reflection principle} for the non-commutative IPS system. A reflection principle for a given proof system $P$ is a  statement that says that if there exists a $P$-proof of a formula $F$ then $F$ is also \textit{true}. The argument becomes rather complicated because we need to prove properties of the PIT algorithm for non-commutative formulas devised by Raz and Shpilka \cite{RS04} within the
restrictive framework of propositional Frege proofs.
\smallskip 

Our goal is then to prove $\neg \phi$ in Frege, given a  non-commutative IPS refutation
$\pi$ of $\phi$.
\smallskip 

\textbf{Step 1: balancing.}
We  first \emph{balance} the non-commutative IPS $\pi$, so that its depth is logarithmic in its size. We observe that the recent construction of Hrube\v s and Wigderson \cite{HW14} for balancing non-commutative formulas with division gates (incurring with at most a  polynomial increase in size) results in a \textit{division-free} formula, when the \textit{initial }non-commutative formula is division-free by itself. Therefore, we can assume that the non-commutative IPS certificate is already balanced (this
step is independent of the Frege system).
\smallskip

\textbf{Step 2: Booleanization.} We then consider our balanced $\pi$, which is a non-commutative polynomial identity \emph{over $GF(2)$, as a Boolean tautology}, by replacing plus gates with XORs and product gates with ANDs. 
\smallskip

\textbf{Step 3:  reflection principle.} We use a reflection principle to reduce the
task of efficiently proving $\neg\phi$ in Frege to the following task: show that any non-commutative formula
identity over $GF(2)$, considered as a Boolean tautology, has a short Frege proof.  

\smallskip

\textbf{Step 4: homogenization.} This is
the \uline{\emph{only}} step that is responsible for the
\textit{quasi-polynomial size increase }in Theorem \ref{thm:intro:Frege_sim_ncIPS}.
More precisely, this increase in size depends on the fact that for the purpose of establishing short Frege proofs for  all non-commutative polynomial identities over $GF(2)$ (considered as Boolean tautological formulas) it is important that the formulas are written as a sum of \textit{homogenous} non-commutative formulas. 

Note that it is not known whether  arithmetic formulas can be turned into a (sum of) homogenous formulas with only a polynomial increase in size (in contrast to the standard efficient homogenization of arithmetic \emph{circuits} by Strassen \cite{Str73} that does allow such a conversion). Nevertheless, Strassen's standard procedure enables us to transform any polynomial-size
arithmetic formula into a sum of homogenous formulas with
only a \textit{quasi-polynomial} increase in size: any formula of size $\poly(n)$ computing
a polynomial $f$ (and thus the degree of $f$ is also polynomial) can be transformed into a sum of homogenous formulas, each having size $n^{O(\log n)}$ and computes the corresponding homogenous part of $f$. (One can show that the same 
also holds for \textit{non-commutative}
formulas.)  

For the purpose of establishing a quasi-polynomial simulation of non-commutative IPS by Frege,
it is sufficient to use the original Strassen's homogenization procedure (as simulated inside Frege; cf.~\cite{HT12}).  However, as the note after Theorem  \ref{thm:intro:Frege_sim_ncIPS} indicates, we show a slightly stronger simulation result, using an efficient Frege simulation of a recent result due to  Raz \cite{Raz13-tensor} who showed how to transform an arithmetic formula into (a sum of) homogenous formulas in a manner which is more efficient than Strassen \cite{Str73}. 
Specifically, in Lemma \ref{homogenous-proof} we show that:
\begin{enumerate}
\item The same construction in \cite{Raz13-tensor} also holds for \textit{non-commutative} formulas; \vspace{-5pt}
\item This construction for non-commutative formulas can be carried out efficiently inside Frege. That is, if $F$ is a non-commutative
formula of size $s$ and depth $d$ computing a homogenous non-commutative polynomial
over $GF(2)$ of degree $r$, then there exists a syntactic homogenous non-commutative formula $F'$ computing the same polynomial and with size $O\left({r+d+1 \choose r}
\cd s\right)$, such that Frege admits a proof of $F\sequiv F'$ of size polynomial
(in $|F'|$).
\end{enumerate}



\smallskip

\textbf{Step 5: short proofs for homogenous non-commutative identities.} Now that we have reduced our
task to the task of showing that every  non-commutative
formula identity over $GF(2)$ (considered as a tautology) has a short Frege proof;
and we have also agreed to first turn (inside
Frege) our non-commutative identities into
\emph{homogenous} formulas (incurring in
up
to a quasi-polynomial increase in the formulas
size)---it remains only to
show how to efficiently prove in Frege homogenous
non-commutative
identities. (Formally, we shall in fact deal
with \emph{syntactic} homogenous formulas.)
 
To this end we essentially construct an efficient  Frege proof of the correctness of the Raz and Shpilka PIT algorithm for non-commutative formulas  \cite{RS04}. This PIT algorithm  uses some basic linear algebraic concepts that might be beyond the efficient-reasoning strength of  Frege. However, since we only need to show the \textit{existence} of short Frege proofs for the PIT algorithm's correctness, we can supply \emph{witnesses} to witness the desired linear algebraic objects needed in the proof (these witnesses
will be a sequence of linear transformations).

A bigger obstacle is that it seems impossible to reason directly inside Frege about the algorithm of \cite{RS04}, since this algorithm first converts a non-commutative formula into an \textit{algebraic branching program }(ABP); but the evaluation of ABPs (apparently) cannot be done with Boolean formulas (and accordingly Frege (apparently) cannot reason about the evaluation of ABPs). The reason for this apparent inability of Frege to reason efficiently about ABP's evaluation is that an ABP is a slightly more ``sequential" object than a formula: an evaluation of an ABP with $d$ layers can be done  by an
iterative matrix multiplication of $d$ matrices---known to be doable with quasi-polynomial
size formulas (or polynomial-size circuits with $O(\log^2 n)$ depth)---while Frege
is a system  that
operates with formulas. 
To overcome this obstacle we show how to perform Raz and Shpilka's PIT algorithm \emph{directly on non-commutative formulas}, without converting the formulas first into ABPs. This technical contribution takes quite a large part of the argument (Sec.~\ref{sec:splitting}). 

We are finally able to prove the following statement, which might be of independent interest:

\begin{theorem}\label{thm:intro:homo-formula}
If a non-commutative \emph{homogeneous} formula $F(\overline x)$ over $GF(2)$ of size $s$ is identically zero, then the corresponding \boolean\ formula $\neg\Fb(\overline x)$ (where $\Fb$ results by replacing $+$ with XOR and $\cd$ with AND in $F(\overline x)$) can be proved with a Frege proof of size at most $s^{O(1)}$.
\end{theorem}

A more detailed \textit{overview} of the proof (specifically, of the proof of Theorem \ref{thm:intro:homo-formula})
appears in Section \ref{sec:remain_proof_overview}.



\subsection{Comparison with previous work}\label{sec:comparison}
Our main characterization of the Frege system is based on a non-commutative version of the IPS system from Grochow and Pitassi \cite{GP14}.  As described above, the non-commutative IPS gives a tighter characterization than the (commutative) IPS in \cite{GP14}, and close to capture almost tightly the Frege system. 

In the original (formula version of the) IPS, proofs are arithmetic formulas, and thus any super-polynomial lower bound on IPS refutations implies $\sf VNP\neq \sf VP_e$, or in other words, that the permanent does not have polynomial-size arithmetic formulas (Joshua Grochow [personal communication]). This shows that proving IPS lower bounds will be considerably difficult to obtain. For the non-commutative IPS, on the other hand, we face a seemingly much favourable situation: an exponential-size lower bound on   non-commutative IPS gives only a corresponding lower bound on non-commutative formulas, for which exponential-size lower bounds are
already known \cite{Nis91}. In other words, exponential-size lower bounds on Frege implies merely---at least in the context of the Ideal Proof System---corresponding lower bounds on non-commutative formulas, a result which is already known. In view
of this, it seems that there is no strong concrete  justification to believe that Frege lower bounds are
beyond  current techniques.
\smallskip

Let us also mention the work in \cite{Tza11-I&C} that dealt with propositional proof systems over non-commutative formulas.  In \cite{Tza11-I&C} the choice was made to define all proof systems as  polynomial calculus-style systems in which proof-lines are written as non-commutative formulas (as well as
the more restricted class of ordered-formulas). 
This meant that the characterization of a proof system in terms of a  \emph{single} non-commutative polynomial is lacking from that work (as well as the consequences we obtained in the current work).

\section{Preliminaries}
For a positive natural number $n$ we use the standard notation $[n]$ for $\{1,\ldots,n\}$.
\begin{definition}[Boolean formulas]
Given a set of input variables $\{x_1,x_2,\ldots\}$ a \emph{Boolean formula} on the input variables
is a rooted finite tree of fan-in at most 2, with edges directed from leaves to the root. We consider the edges coming into nodes as \emph{ordered}.\footnote{This
is not important in general, but for Frege proofs it is in fact implicit that propositional
formulas are ordered.}
Internal nodes are labeled with the Boolean gates OR, AND and NOT, denoted $\lor,\land, \neg$, respectively, where the fan-in of $\lor$ and $\land$ is two and the fan-in of $\neg$ is one. The leaves are labeled either with input variables or with $0,1$ (identified with the truth values {\rm \textsf{false}} and {\rm \textsf{true}}, resp.).
The entire formula computes the function computed by the gate at the root. Given a formula $F$, the
\textbf{size} of the formula is the number of Boolean gates in $F$, denoted $|F|$.
\end{definition}

Given a pair of Boolean formulas $ A$ and $ B $ over the variables $x_1,\ldots,x_n$, we denote by $A[B/x_i]$ the formula $ A $ in which \textit{every occurrence of $ x_i $ } in $A$ is substituted by the formula $ B $.

We use the symbol $\equiv$ to denote \emph{logical equivalence} and we use the symbol
$A\sequiv B$ to denote $(A\rightarrow B) \land (B\rightarrow A)$.
\subsection{The Frege proof system}

As outlined in the introduction, a Frege proof system is any standard propositional
proof system for proving propositional tautologies having finitely many
axiom schemes and deduction rules, and where proof-lines are written as Boolean\ formulas. The \emph{size} of a Frege proof is the number of symbols it takes to write down the proof, namely the total of all the  formula sizes appearing in the proof.
Let us define Frege proofs in a more formal way.  
\begin{definition} [Frege (derivation) rule] A \emph{Frege rule} is a sequence of propositional formulas $A_0(\overline x),\ldots,A_k(\overline x)$, for $k \le 0$, written as $\frac{A_1(\overline x), \ldots,A_k(\overline x)}{A_0(\overline x)}$. In case $k=0$, the Frege rule is called an \emph{axiom scheme}. A formula $F_0$ is said to be \emph{derived by the rule} from $F_1,\ldots,F_k$ if $F_0,\ldots,F_k$ are all substitution instances of $A_1,\ldots,A_k$, for some assignment to the $\overline x$ variables (that is, there are formulas $B_1,\ldots,B_n$ such that $F_i = A_i[B_1/x_1,\ldots,B_n/x_n]$, for all $i=0,\ldots,k$). The Frege rule is said to be \emph{sound} if whenever an assignment satisfies the formulas  $A_1,\ldots,A_k$ above the line, then it also satisfies the formula $A_0$ below the line.
\end{definition}

\begin{definition} [Frege proof] Given a set of Frege rules, a \emph{Frege proof} is a sequence of Boolean formulas such that every formula is either an axiom or was derived by one of the given Frege rules from previous formulas. If the sequence terminates with the Boolean formula $A$, then the proof is said to be a \emph{proof} of $A$. The \textbf{size} of a Frege proof is  the sum of all formula sizes in the proof.
\end{definition}

A proof system is said to be \textit{sound} if it admits proofs of only tautologies.
A proof system is said to be \emph{implicationally complete} if for all set of formulas $S$, if $S$ semantically implies $F$, then there is a proof of $F$ using (possibly) axioms from $S$. 
\begin{definition} [Frege proof system]\label{def:Frege_system} Given a set $P$ of sound Frege rules, we say that $P$ is a \emph{Frege proof system} if $P$ is implicationally complete.
\end{definition}

Note that a Frege proof is always sound since the Frege rules are assumed to be sound.
Frege is also complete (that is, can prove all tautologies), by implicational
completeness. We do not need to work with a specific Frege proof system, since a basic result in proof complexity by Reckhow \cite{Rec76:PhD} states that every two Frege proof systems, \textit{even with different propositional connectives}, are polynomially equivalent. For concreteness the reader can think of Schoenfield's system from the introduction, noting it is indeed a Frege system. \smallskip

The problem of demonstrating super-polynomial size lower bounds on propositional Frege proofs  asks whether there is a family $(F_n)_{n=1}^\infty$ of propositional tautological formulas for which there is no polynomial $p$ such that the minimal Frege proof size of $F_n$ is at most $p(|F_n|)$, for all $n\in \mathbb{Z}^+$.



\subsection{Preliminary  algebraic models of computation and proofs }
Here we define arithmetic formulas (both commutative and non-commutative)
as well as the algebraic propositional proof system
Polynomial Calculus over Formulas ($\mathcal{F\mbox{\rm-}PC}$) introduced by Grigoriev and Hirsch \cite{GH03}.
 
\begin{definition}[Non-commutative formula]\label{def:nonc_formula}
Let $ \F $ be a field and $\{x_1,x_2,\ldots\} $ be (algebraic) variables. A \emph{non-commutative arithmetic formula} (or \emph{non-commutative formula} for short) is a finite (ordered)
labeled tree, with edges directed from the leaves to the root, and with fan-in at most two, such that there is an order on the edges coming into a node: the first edge is called the \emph{left} edge and the second one the \emph{right} edge. Every leaf of the tree (namely, a node of fan-in zero) is labeled either with an input variable $ x_i $ or a field element. Every other node of the tree is labeled either with $ +$ or $\times $ (in the first case the node is a plus gate and in the second case a non-commutative
product gate). We assume that there is only one node of out-degree zero, called \emph{the root}. 
\end{definition}
A non-commutative formula \emph{computes} a non-commutative polynomial in $ \F\langle x_1,\ldots,x_n\rangle  $ in the following way. 
A leaf computes the input variable or field element that labels it. A plus gate computes the sum of polynomials computed by its incoming nodes. A product gate computes the \emph{non-commutative} product of the polynomials computed by its incoming nodes according to the order of the edges. (Subtraction is obtained using the constant $ -1$.) The output of the formula is the polynomial computed at the root. The \textbf{\emph{depth} }of a formula is the maximal length of a path from the root to the leaf. The \textbf{\textit{size}} of a non-commutative  formula $ F $ is the total number of internal nodes (i.e., all nodes except the leaves) in its underlying tree, and is denoted similarly to the Boolean case by $|F|$.

The definition of (a commutative) arithmetic formula is almost identical:
\begin{definition}[(Commutative) arithmetic formula]
An \emph{arithmetic formula} is defined in a similar way to a non-commutative formula, except that we ignore the order of multiplication (that is, a product node does not have order on its children and there is no order on multiplication when defining the polynomial computed by a formula).
\end{definition}

Substitutions of  non-commutative formulas into other non-commutative formulas
are defined
and denoted similarly to  substitutions in Boolean formulas.
\QuadSpace


Note that we consider arithmetic formulas as syntactic objects. For example, $x_1+x_2$ and $x_2+x_1$ are different formulas. Furthermore, in the proof system \FPC\  defined below they should be  \emph{derived} from each other via an explicit application of  a rewrite rule.

\subsubsection{Polynomial calculus over formulas $\mathcal{F\mbox{\rm-}PC}$}

The \emph{polynomial calculus over formulas} system, denoted $\mathcal{F\mbox{\rm-}PC}$, was introduced
by Grigoriev and Hirsch \cite{GH03}. This system  operates with (commutative) arithmetic formulas (as purely syntactic terms). \FPC\ is a refutation system: an \FPC\ refutation  establishes that a collection of polynomials has no 0-1 roots.
We can also treat \FPC\
as a \emph{proof} system for propositional tautologies: for every Boolean tautology $T$, $\trn(\neg T)$ (Definition \ref{def:trn_1}) is a polynomial that does not have
a 0-1 root, and therefore, an \FPC\ refutation of $\trn(\neg T)$ can be considered
as an \emph{\FPC\ proof of the tautology $T$}. 

\begin{definition}[Rewrite rule]
A \emph{rewrite rule} is a pair of formulas $ f,g $ denoted $ f \rightarrow g $.
Given a formula $ \Phi $, an \emph{application of a rewrite rule $ f \rightarrow g $ to $\Phi $}
is the result of replacing at most one occurrence of $ f $  in  $ \Phi $  by  $ g $ (that is,
substituting a subformula $ f $  inside $\Phi $ by the formula $ g $).
We write $ f \leftrightarrow g $ to denote the pair of rewriting rules $ f \rightarrow g $ and $ g \rightarrow f $.
\end{definition}

\begin{definition}[$ \mathcal{F\mbox{\rm-}PC} $ \cite{GH03}]\label{def:F-PC}
 Fix a field $ \F $. Let $ F:=\set{f_1,\ldots,f_m} $ be a collection of \emph{formulas}\footnote{Note here that we are talking about formulas (treated as syntactic terms). Also notice that all the formulas in $ \mathcal{F\mbox{\rm-}PC} $ are considered  as commutative formulas computing (commutative) polynomials, though, because the formulas are merely
syntactic terms we have an order on  children of internal nodes, and in particular
children of product gates are ordered.}
 computing polynomials from $ \F[x_1,\ldots,x_n] $.
 Let the set of axioms be the following formulas:
 \begin{description}
   \item[\quad Boolean axioms]\qquad
   $ x_i\cdot(1-x_i)\,, \qquad  \mbox{ for all $\,1\le i\le n\,$.}$
 \end{description}
    A sequence $\pi =(\Phi_1,\ldots,\Phi_\ell)$ of formulas computing polynomials from
   $\F[x_1,\ldots, x_n]$\, is said to be \textbf{an $ \mathcal{F\mbox{\rm-}PC} $ proof of $ \Phi_\ell $ from $F$},
   if for every $i\in[\ell]$ we have one of the following:
   \begin{enumerate}
      \item $\Phi_i = f_j\,$, for some $j\in[m]$;
      \item $\Phi_i$ is a Boolean axiom;
      \item $\Phi_i$ was deduced by one of the following inference rules from previous proof-lines
            $\Phi_j, \Phi_k\,$, for $j,k<i$:
            \begin{description}
                \item[\quad Product]
                \[
                    \frac{\Phi}{x_r\cd \Phi}\ ,   \qquad \qquad \mbox{for $ r\in[n]$}\,.
                \]
                \item[\quad Addition]
                \[
                    \frac{{\Phi\quad \quad \Theta}}{{a\cd \Phi + b\cd \Theta}}\ ,\qquad
                            \,{\mbox{for}}\  a,b \in \F\,.
                \]
            \end{description}
            (Where $\Phi, x_r\cd\Phi, \Theta, a\cd\Phi, b\cd\Theta $ are \emph{formulas}
            constructed as displayed; e.g.,
            $ x_r\cd\Phi $ is the formula with product gate at the root having the
            formulas $ x_r $ and $ \Phi $
            as children.)\footnotemark
      \item $\Phi_i $ was deduced from previous proof-line $\Phi_j $, for $ j<i $, by one of the following \emph{rewriting rules}
      expressing the polynomial-ring axioms (where $f,g,h$ range over all arithmetic formulas computing polynomials in $\F[x_1,\ldots,x_n]$):
             \begin{description}
                 \item[Zero rule]
                 $ 0\cd f \leftrightarrow 0 $

                \item[Unit rule]
                $  1\cd f \leftrightarrow f$

                \item[Scalar rule]
                $ t \leftrightarrow \alpha $, where $t $ is
                a formula containing no variables (only field $ \F $  elements) that computes the constant $ \alpha\in\F $.

                \item[Commutativity rules]
                $ f + g \leftrightarrow g + f \,$, \qquad $ f\cd g \leftrightarrow g\cd f$
                \item[Associativity rule]
                $ f + (g+h) \leftrightarrow (f+g)+h \,$,  \qquad    $ f\cd(g\cd h) \leftrightarrow (f\cd g)\cd h $
                \item[Distributivity rule]
                $ f \cd(g+h) \leftrightarrow (f\cd g)+(f\cd h) $
             \end{description}
    \end{enumerate}
(The semantics of an $ \mathcal{F\mbox{\rm-}PC} $ proof-line $ p_i$ is the polynomial equation $p_i=0$.)

An \emph{$ \mathcal{F\mbox{\rm-}PC} $ refutation of} $F$ is a proof of the formula $\;1$ from $F$.
The \textbf{size} of an  $ \mathcal{F\mbox{\rm-}PC} $ proof $\pi$ is defined as the total size of all
formulas in $\pi$ and is denoted by $|\pi|$.
\end{definition}
\footnotetext{In \cite{GH03} the product rule of $ \mathcal{F\mbox{\rm-}PC} $ is defined so that one can derive $ \Theta\cd\Phi $ from $ \Phi $,         where $ \Theta $ is any formula, and not just a variable. However, it is easy to show that
the definition             of $ \mathcal{F\mbox{\rm-}PC} $ in \cite{GH03} and our Definition \ref{def:F-PC}            polynomially-simulate each other.}

\begin{definition}[Tree-like $ \mathcal{F\mbox{\rm-}PC} $]
A system $ \mathcal{F\mbox{\rm-}PC} $ is a  \emph{tree-like} $ \mathcal{F\mbox{\rm-}PC} $ if
 every derived arithmetic formula in the proof system  is used only once (and if it is needed again, it must be derived once more).
\end{definition}
%

%
%


For the purpose of comparing the relative complexity of different proof systems we have the concept
of a \textbf{\emph{simulation}}. Specifically, we say that a propositional proof system $P$ \emph{polynomially simulates} another propositional proof system $Q$ if there is a polynomial-time computable function $f$ that maps $Q$-proofs to $P$-proofs of the same tautologies (if $P$ and $Q$ use different representations for tautologies,
we fix a translation (such as $\trn(\cd)$) from one representation to the other).
In case $f$ is computable in time $t(n)$ (for $n$ the input-size), we say that $P$ \emph{~$t(n)$-simulates}
$Q$. Specifically, if $t(n)=n^{O(\log n)}$ we say the simulation is \emph{quasi-polynomial}. We say that $P$ and $Q$ are \emph{polynomially equivalent} in case $P$ polynomially
simulates $Q$ and $Q$ polynomially simulates $P$. (Our simulations will always be formally  $t(n)$-simulations, though we might not always state explicitly that the map $f$, from $Q$-proofs to $P$-proofs is efficiently computable, and only show the \emph{existence} of a $P$-proof
whose size is proportional to the corresponding $Q$-proof.)

\para{Tree-like \FPC\ polynomially simulates Frege.}

Grigoriev and Hirsch showed the following:

\begin{theorem}[\cite{GH03}]\label{tree-like_F-PC=Frege}
Tree-like $\mathcal{F\mbox{\rm-}PC}$ polynomially simulates Frege. More precisely, for
every propositional tautology T, if T has a polynomial-size Frege proof then there is a polynomial-size tree-like \FPC\ proof of $\hbox{ \rm \trn}(\neg T)$ (over $\Z_q$, for $q$ a prime, or $\mathbb Q$).
\end{theorem}

Let us shortly explain how Grigoriev and Hirsch  \cite{GH03} obtained a simulation of Frege by \emph{tree-like}
\FPC\ (in contrast to simply (dag-like) \FPC), as this is not an entirely trivial
result (and which, in turn, is important to understand our  simulation). Indeed, this simulation depends crucially on a somewhat surprising result of Kraj\'{i}\v{c}ek who showed that tree-like Frege and (dag-like) Frege are \emph{polynomially equivalent} \cite{Kra95}:
\begin{theorem*}[\cite{Kra95}]
  Tree-like Frege proofs polynomially simulate Frege proofs.
\end{theorem*}


Grigoriev and Hirsch show that (Theorem 3 in \cite{GH03})  $\mathcal{F\mbox{\rm-}PC}$ polynomially simulates  Frege. Then, by inspection of this simulation, one can observe
that tree-like Frege proofs are simulated by tree-like \FPC\ proofs (which is sufficient
to conclude the simulation due to the theorem above), namely:
\begin{lemma*}[\cite{GH03}]
Tree-like $\mathcal{F\mbox{\rm-}PC}$ polynomially simulates tree-like Frege.
\end{lemma*}

\section{Non-commutative ideal proof system polynomially simulates Frege}

Here we show that the non-commutative IPS polynomially simulates
Frege. \begin{theorem}[restatement of Theorem \ref{thm:intro:ncIPS_sim_Frege}]
\label{thm:ncIPS_sim_Frege}
The non-commutative IPS refutation system (when refutations are written as non-commutative formulas) polynomially simulates the Frege system. More precisely, for every propositional tautology $T$, if $T$ has a polynomial-size Frege proof then there is a non-commutative IPS refutation of $\trn(\neg T)$ (over $\Z_p$ for a prime $p$, or $\mathbb Q$) of polynomial size.
\end{theorem}


%

Recall that Raz and Shpilka \cite{RS04} gave a deterministic polynomial-time
PIT algorithm for non-commutative formulas (over any field):

\begin{theorem}[PIT for non-commutative formulas \cite{RS04}]\label{thm:RS04-PIT}
There is a deterministic polynomial-time algorithm that decides whether a given noncommutative formula over a field $ \F $ computes the zero polynomial $ 0 $.\footnote{We assume here that the elements of   $ \F $  have an efficient representation and the field
operations are efficiently computable (e.g., the field of rationals).}
\end{theorem}

Now, since we write refutations as non-commutative formulas we can use the theorem above to check in \textit{deterministic} polynomial-time the correctness of non-commutative IPS refutations, obtaining:

\begin{corollary}[restatement of Corollary \ref{cor:intro:Cook-Reckhow}]
The non-commutative IPS is a sound and complete Cook-Reckhow refutation system. That is, it is a sound and complete refutation system for unsatisfiable propositional formulas in which refutations can be checked for correctness in deterministic polynomial-time.
\end{corollary}
To prove Theorem \ref{thm:ncIPS_sim_Frege}, we will show in Section \ref{sec:ncIPS_sim_tree_FPC} that the
non-commutative IPS polynomially-simulates tree-like \FPC\ (Definition
\ref{def:F-PC}), which sufficed to complete the proof, due to Theorem
\ref{tree-like_F-PC=Frege}.

\newcommand{\cv}{\ensuremath{\mathbf C}}
\newcommand{\fv}{\ensuremath{\mathbf F}}

\subsection{Non-commutative IPS polynomially simulates tree-like $\mathcal{F\mbox{\rm-}PC}$}\label{sec:ncIPS_sim_tree_FPC}

For convenience, let $C_{i,j}$ denote the commutator axiom $x_i\cdot x_j-x_j\cdot x_i$, for $i,j\in [n], i\ne j$, and let $\cv$ denote the vector of all the $C_{i,j}$
axioms.  
When we write $P\cdot Q-Q\cdot P$ where $P,Q$ are formulas (e.g., $x_i$ and $x_j$, resp.), we mean $((P\cdot Q)+(-1\cdot(Q\cdot P)))$.
\begin{theorem}\label{ncIPS2Frege}
Non-commutative IPS polynomially simulates tree-like $\mathcal{F\mbox{\rm-}PC}$ (Definition \ref{def:F-PC}). Specifically, if $\pi$ is a tree-like \FPC\ proof of a tautology
$T$ then there is a non-commutative IPS refutation of $\trn(\neg T)$ of size polynomial
in $|\pi|$.  
\end{theorem}

\begin{proof}
Let  $F_1, \ldots, F_m$ be arithmetic formulas over the variables  $x_1,\ldots, x_n$. We denote by $\fv$ the vector $(F_1,\ldots,F_m)$. Since an arithmetic formula is a syntactic term in which the children of gates are ordered we can treat a (commutative) arithmetic formula as a  \emph{non-commutative} arithmetic formula by taking the \textit{order }on the children of products gates to be the order of non-commutative multiplication.

Suppose $\mathcal{F\mbox{\rm-}PC}$ has a $\poly(n)$-size tree-like refutation $\pi:=(L_1, \ldots, L_k)$ of the $F_i$'s (i.e., a proof of the polynomial $1$ from $F_1,\ldots, F_m$), where each $L_j$ is an arithmetic formula.  We construct a corresponding non-commutative IPS refutation of the $F_i$'s from this \FPC\ tree-like refutation. The following
lemma suffices for this purpose:\begin{lemma}\label{lem-IPSproperty}
For every $i\in[k]$, there exists a non-commutative formula $\phi_i$ such that
 \begin{enumerate}
\item $\phi_i(\overline x, \overline 0)=0$;
\item $\phi_i(\overline x, \fv,\cv)=L_i$\,;


\item $|\phi_i|\leq \proofbound{A_i}$, where $A_i\subset [k]$ are  the indices of the \FPC\ proof-lines involved in deriving $L_i$.

For example, if $L_i$ is derived by $L_\a$ and  $L_\a$ is derived by $L_\b$ for some $\b<\a<i\in[k]$, then we say that $\a,\b$ are both involved in deriving $L_i$. In
other words, the lines involved in deriving a proof-line $L_i$ are all the proof-lines in the sub-tree of $L_i$ when we consider the underlying graph of the (tree-like) proof as a tree.
\end{enumerate}

\end{lemma}
Note that if the lemma holds, then $\phi_k$ is a non-commutative IPS proof because it has the property that $\phi_k(\overline x, \overline 0)=0$ and $\phi_k(\overline x, \fv,\cv)=L_k=1$. And its size is bounded by $\proofbound{A_k}\leq \proofbound{[k]}\leq O(|\pi|^4).$ 
\medskip
\proof
We construct $\phi_i$ by induction on the length $k$ of the refutation $\pi$. That is, for $i$ from $1$ to $k$, we construct the non-commutative formula $\phi_i(\overline x, \overline y)$  according to $L_i$, as follows:

\Base 
 $L_i$ is an axiom $F_j$ for some $j\in[m]$.

Let $\phi_i:=y_j$. Obviously, $\phi_i(\overline x,0)=0, \phi_i(\overline x,\fv,\cv)=F_j=L_i$ and $|\phi_i|=1 \leq \ |L_i|^4$.

\induction 


\case 1  $L_i$ is derived from the addition rule $L_i=aL_j+bL_{j'}$, for $j,j'<i$. Put $\phi_i:=a\phi_j+b\phi_{j'}$ where $a,b\in \F$. Thus, $\phi_i(\overline x,0)=a\phi_j(\overline x,0)+b\phi_{j'}(\overline x,0) = 0,\,\phi_i(\overline x,\fv,\cv)=aL_{j}+bL_{j'}=L_i$ and $|\phi_{i}|=|\phi_j|+|\phi_{j'}|+3\leq \proofbound{A_j}+\proofbound{A_{j'}}+3\leq \proofbound{A_i}$ (where the right most inequality holds since $\pi$ is a \textit{tree-like} refutation and hence $A_j \cap A_{j'}=\emptyset$).

\medskip

\case 2 $L_i$ is derived from the product rule $L_i=x_r\cdot L_{j}$, for $r\in[n]$
and $j<i$. Put $\phi_i:=(x_r\cdot \phi_j)$. Then $\phi_i(\overline x,0)=x_r\cd\phi_j(\overline x,0) = 0,  \phi_i(\overline x,\fv,\cv)=x_r\cd L_j=L_i$ and $|\phi_{i}|=|\phi_j|+2\leq \proofbound{A_j}+2\leq \proofbound{A_i}$.

\medskip 
\case 3
$L_i$ is derived from $L_j$, for \(j<i\), by a \emph{rewriting} rule
which is not the commutative rule of multiplication ($f\cd g\leftrightarrow g\cd
f$). Let $\phi_i:=\phi_j$. The non-commutative $\phi_i$ trivially satisfies the properties  claimed since all the rewriting  rules (excluding the commutative rule of multiplication)  express the non-commutative polynomial-ring axioms,  and thus cannot change the polynomial computed by a non-commutative formula. And $|\phi_i|=|\phi_j|\leq \proofbound{A_i}$.

\medskip
\case 4
$L_i$ is derived  from $L_j$, for $j<i$, by a single application of the commutative rule of multiplication. Then by Lemma \ref{LemmaLiLj} below, we can construct a non-commutative formula $\commF$ such that $\phi_i:=(\phi_j+\commF)$ satisfies the desired properties (stated in Lemma~\ref{lem-IPSproperty}).

\end{proof}

%
%
%
%
%
%
%
%
%

\begin{lemma} \label{LemmaLiLj}
Let $L_i, L_j$ be non-commutative formulas, such that  $L_i$ can be derived from $L_j$ via the commutative rule of multiplication $f\cd g\leftrightarrow g\cd f$. Then there is a non-commutative formula $\commF(\overline x, \overline y)$ in variables $\{x_\ell, y_{\a,\b},\; \ell\in[n],  \a<\b \in [n] \},$ such that:
\begin{enumerate}
\item $\commF(\overline x, \overline 0)=0$;
\item $\commF(\overline x,\cv)=L_i-L_j$;
\item $\left|\commF\right|\leq \left|L_i\right|^2\left|L_j\right|^2$.
\end{enumerate}
\end{lemma}

\begin{proof}
We define the non-commutative formula $\commF$ inductively  as follows:
\begin{itemize}
        \item If $L_i=(P\cdot Q)$, and $L_j=(Q\cdot P)$, then  $\commF$ is defined to be the formula constructed in Lemma~\ref{LemmaPQQP} below.
        \item If $L_i=(P\cdot Q)$, $L_j=(P'\cdot Q')$.

         \case 1 If $P=P'$, then let $\commF:=(P\cdot \phi_{Q,Q'})$.

         \case 2 If $Q=Q'$, then let $\commF:=(\phi_{P,P'}\cdot Q)$.

        \item If $L_i=(P+ Q)$, $L_j=(P'+ Q')$.

         \case 1 If $P=P'$, then let $\commF= \phi_{Q,Q'}$.

         \case 2 If $Q=Q'$, then let $\commF=\phi_{P,P'}$.

    \end{itemize}

    By induction, the construction satisfies the desired properties.
\end{proof}

\begin{lemma} \label{LemmaPQQP}
For any pair  $P, Q$ of two non-commutative formulas there exists a non-commutative formula $F$ in variables $\{x_\ell, y_{i,j}, \; \ell\in[n],  i<j\in [n]\}$ such that:
\begin{enumerate}
\item $F(\overline x, \overline 0)=0$;
\item $F(\overline x,\cv)=P\cdot Q-Q\cdot P$;
\item $\left|F\right|=\left|P\right|^2\left|Q\right|^2$.
\end{enumerate}
\end{lemma}
\begin{proof}
Let $s(P,Q)$ denote the smallest size of $F$ satisfying the above properties. We will show that $s(P,Q)\le \left|P\right|^2\cdot \left|Q\right|^2$ by induction on $\max(\left|P\right|,\left|Q\right|)$.

\Base  $\left|P\right|=\left|Q\right|=1$.

In this case both $P$ and $Q$ are constants or variables, thus $s(P,Q)=1\le \left|P\right|^2\left|Q\right|^2$.
\medskip

In the following induction step, we consider the case where $\left|P\right| \ge \left|Q\right|$ (which is symmetric for the case $\left|P\right| < \left|Q\right|$).

\Induction Assume that   $\left|P\right| \ge \left|Q\right|$\,. 

\case 1 The root of $P $ is addition.

Let $P=(P_1+P_2)$. We have (after rearranging):
$$
P\cdot Q-Q\cdot P=((P_1\cdot Q-Q\cdot P_1)+(P_2\cdot Q-Q\cdot P_2))
$$
By induction hypothesis, we have $s(P,Q)\le s(P_1,Q)+1+s(P_2,Q)\le \left|P_1\right|^2\left|Q\right|^2+1+ \left|P_2\right|^2\left|Q\right|^2
\le (\left|P_1\right|+\left|P_2\right|+1)^2\left|Q\right|^2
= \left|P\right|^2\cdot \left|Q\right|^2$.

\case 2 The root of $P$ is a product gate.

Let $P=(P_1\cdot P_2)$. By rearranging:
$$
P\cdot Q-Q\cdot P=((P_1\cdot( P_2\cdot Q-Q\cdot P_2))+((P_1\cdot Q-Q\cdot P_1)\cdot P_2))
$$
By induction hypothesis, we have $s(P,Q)=\left|P_1\right|+1+s(P_2,Q)+1+s(P_1,Q)+1+\left|P_2\right|
\le \left|P_1\right|+1+\left|P_2\right|^2\left|Q\right|^2+1+\left|P_1\right|^2\left|Q\right|^2+1+\left|P_2\right|
\le (\left|P_1\right|+\left|P_2\right|+1)^2\left|Q\right|^2
= \left|P\right|^2\cdot \left|Q\right|^2$.
\end{proof}

\section{Frege quasi-polynomially simulates non-commutative IPS}
\label{sec:5-the whole converse simulation}


In this long section we prove  Theorem \ref{thm:Frege_sim_ncIPS} stating that the Frege system  quasi-polynomially simulates  the non-commutative IPS (over $GF(2)$). Together with Theorem \ref{thm:ncIPS_sim_Frege}, this  gives a new characterization (up to a quasi-polynomial increase in size) of propositional Frege proofs as non-commutative arithmetic formulas.

%

We use the notation in Section \ref{sec:intro:ncIPS} as follows: for a clause $\kappa_i$ in a CNF $\phi=\kappa_1\wedge \ldots \wedge \kappa_m$, we denote by $Q_i^{\phi}$ the non-commutative formula  translation
$\trn'(\kappa_i)$ of the clause $\kappa_i$ (Definition \ref{def:trn_2}). Thus,
$\neg x$ translates to $x$, $x$ translates to $1-x $ and $f_1\cdots
f_r$ translates to $\prod_i \trn'(f_i)$ (considered as a tree
of product gates with $\trn'(f_i)$ as leaves), and where the formulas are over $GF(2)$ (meaning that $1-x$ is in fact $1+x$). Recall that  this way, for every 0-1 assignment (when we
identify \textsf{true} with 1 and \textsf{false} with 0), $Q_i^{\phi}=0$
iff $\kappa_i$ is \textsf{true}.   

\begin{theorem}[Second main theorem;
Restatement of Theorem \ref{thm:intro:Frege_sim_ncIPS}]
\label{thm:Frege_sim_ncIPS}
For a 3CNF $\phi=\kappa_1\wedge \ldots \wedge \kappa_m$ where $Q_1^{\phi},\ldots, Q_m^{\phi}$  are the corresponding polynomial equations for the clauses,  if there is a \NCFIPS\ refutation of size $s$ of    $Q_1^{\phi},\ldots, Q_m^{\phi}$ over $GF(2)$,  then there is a Frege proof of size $s^{O(\log s)}$ of $\neg\phi$. 
\end{theorem}

As mentioned in the introduction, it will be evident that our proof in fact establishes
 a slightly tighter simulation of the non-commutative IPS by Frege. Specifically, if the degree of the non-commutative IPS refutation is $r$ and its formula depth is $d$, then there is a Frege proof of $\neg\phi$ with size ${\rm poly}\!\left({d+r+1 \choose r} \cd s\right)$. This will follow from our efficient simulation within Frege
of Raz' \cite{Raz13-tensor} homogenization construction (Lemma \ref{homogenous-proof}).
Nevertheless, for simplicity we shall always assume that the depth $d$ of the non-commutative
IPS refutation formula is logarithmic in the size
$s$ and that the degree $r$ of the refutation is at most $s+1$,  and thus will not take care to explicitly establish the dependence of the 
simulation on the parameters $d$ and $r$.   
\medskip

The rest of the paper is dedicated to proving Theorem \ref{thm:Frege_sim_ncIPS}.

\subsection{Balancing non-commutative formulas}

First we show that a non-commutative formula of size $s$ can be balanced to  an equivalent formula of depth $O(\log s)$, and thus we can assume that the non-commutative IPS certificate is already given as
a balanced formula (this is needed for what follows).  Both the statement of the balancing construction
and its proof are similar  to Proposition 4.1 in Hrube\v s and Wigderson \cite{HW14} (which in turn is similar to the case of   commutative formulas with division gates in Brent \cite{Bre74}).  (Note that a formula of a logarithmic depth (in the number
of variables) must have a polynomial-size (in the number of variables).)

\begin{lemma}\label{lem:balance-non-commutative-formula}
Assume that a non-commutative polynomial $p$ can be
computed by a formula of size $s$. Then $p$ can be computed by a formula of depth
$O(\log s)$ (and hence of polynomial-size when $s$ is polynomial in the number of variables).
\end{lemma}

\newcommand{\bln}[1]{\ensuremath{{\widehat{#1}}}}
\newcommand{\Fh}{\ensuremath{\bln F}}
\newcommand{\tobe}{\leftarrow}

\begin{proof}
The proof is almost identical to Hrube\v s and Wigderson's proof of Proposition 4.1 in \cite{HW14}, which deals with rational functions and allows formulas with division gates.
Thus, we only outline the argument in \cite{HW14} and argue that if the given formula does not have division gates, then the new formula obtained by the balancing construction will not contain any division gate as well.

\renewcommand{\g}{\textit g}
\newcommand{\Fv}{F_\g}
\newcommand{\Fvout}{F[z/\g]}
\newcommand{\FvBack}{F[z\tobe G]}
\newcommand{\Fhv}{\Fh_\g}
\newcommand{\Fhvout}{\bln(F[z/\g])}

\begin{notation*}
Let $F$ be a non-commutative formula and let $\g$ be a gate in $F$. We denote by $F_{\g}$ the subformula of $F$ with the root being $\g$ and by $\Fvout$ the formula obtained by replacing $F_\g$ in $F$ by the variable $z$. We denote by $\widehat F, \widehat F_g$
the non-commutative \emph{polynomials} in $\F\langle \overline X\rangle$
computed by $F$ and $F_g$, respectively.  
\end{notation*}

We simultaneously prove the following two statements by induction on $s$, concluding the lemma:

 \smallskip

\ind\textbf{Inductive statement}:
\textit
{If $F$ is a non-commutative formula of size $s$, then for sufficiently large $s$ and suitable constants $c_1,c_2>0$, the following hold:
  \begin{itemize}
    \item[(i)] $\widehat F$ has a non-commutative formula of depth at most  $c_1 \log s+1$\,;
    \item[(ii)] if $z$ is a variable occurring at most once in $F$, then:
    $$\widehat F = A\cd z \cd B+ C,$$
where $A, B,C$ are non-commutative polynomials  that do not
contain $z$, and each can be computed by a non-commutative formula
of depth at most
$ c_2 \log s$.
\end{itemize}
}

\Base  $s=1$. In this case there is one gate $\g$ connecting two variables or constants. Thus, (i) in
the inductive statement can be obtained immediately as it is already computed by a formula of depth $1=\log s + 1$. As for (ii), note that in the base case, $F$ is a formula with only one gate $\g$.  Assuming that $z$ is a variable occurring only once in $F$,  it is easy  to construct non-commutative formulas $A,B,C$ so that $\widehat F =A\cd z \cd B+ C$ for which the conditions in (ii) hold as follows:
\smallskip

\case 1 if $\g$ is a plus gate connecting the variable $z$ with a variable or constant $x\neq z$, then we can write $F$ as $1\cd z\cd 1+x$.

\case 2 if $\g$ is a product gate  connecting $z$ with $x$ (for $z\neq x$, and in this order),  then we can
write $F$ as $1\cd z\cd x+0$.

\case 3 if $\g$ is a product gate connecting $x$ with $z$ (for $z\neq x$, and in this order), then we can write
$F$ as $x\cd z\cd 1 +0$.

\induction
 (i) is established  (slightly informally) as follows. Find a gate $\g$
in $F$ such that both $\Fv$ and $\Fvout$ are small (of size at most $2s/3$, and where $z$ is a new variable that does not occur in $F$). Then, by applying induction hypothesis on $\Fvout$, there exist formulas $A,B,C$ of small depth such that $\widehat {F[z/g]}=A\cd z \cd B + C$. Thus, $\widehat F:= A\cd \widehat {F_g} \cd B + C.$

To prove (ii), find an appropriate gate $\g$ on the path between $z$ and the output of
$F$  (an \textit{appropriate}  $\g$ is a gate $g$ such that $F[z_1/\g]$
 and $\Fv$
are both small (of size at most $2s/3$), where $z_1$ is a new variable not occurring
in $F$). Use the inductive assumptions to write:
$$
\widehat  F[z_1/\g] = A_1\cd z_1\cd B_1 + C_1
\hbox{~~and~~}  \widehat F_g  = A_2\cd z \cd B_2 + C_2
$$
and compose these expressions to  get 
$$\Fh =A_1\cd (A_2\cd z \cd B_2+C_2) \cd B_1+C_1=A'\cd z \cd B'+ C',$$
where $A'=A_1 \cd A_2$, $B'=B_2\cd B_1,C'=A_1\cd C_2 \cd B_1+C_1$. 

 It is clear that the respective depth of $A', B'$ and $C'$ are all at most $c_2\log (2s/3)+2\le c_2\log s $ when $s$  is sufficiently large.  

To finish the proof of (ii), it suffices to show  that $A',B',C'$  do not contain the variable $z$. It is enough to prove that $A_1,B_1,C_1, A_2,B_2,C_2$ do not contain $z$. Notice that $\Fv$ contains  $z$ and  $z$ is a variable occurring at most once in $F$. Therefore $\widehat {F[z_1/\g]}$ does not contain the variable $z$, which means that
both $A_1,B_1,C_1$  do not contain $z$. Moreover, by induction hypothesis, we know that $A_2,B_2,C_2$  do not contain $z$. Therefore, we  conclude that $A',B',C'$  do not contain $z$.
\end{proof}

As a consequence of Lemma \ref{lem:balance-non-commutative-formula}, in what follows, without loss of generality we will assume that $F$ is given already in a balanced form, namely has depth $O(\log s)$ and size $s$. 

\subsection{The reflection principle}\label{sec:refl_prncple}

Here we show that the existence of a \NCFIPS\ refutation  of the CNF $\phi$ with size $s$ and depth $O(\log s)$  implies the existence of a Frege proof of $\neg\phi$ with size $s^{O(\log s)}$ (we use the same notation as in the beginning of Section \ref{sec:5-the whole converse simulation}).
 This is done by proving a \textit{reflection principle} for the non-commutative IPS  inside Frege. As mentioned in the introduction, informally, a reflection principle for a given proof system $P$ is a  statement asserting that if a formula is \emph{provable }in $P$ then the formula is also \textit{true}. Thus, suppose we have a short Frege
proof of the following reflection principle for $P$: 
\begin{center}
        ``($[\pi]$ is a $P$-proof of 
                                $[T])\,\longrightarrow\,T$'',
\end{center}
where $[T]$ and $[\pi]$  are some reasonable encodings of the tautology
$T$ and its $P$-proof $\pi$, respectively. Then, it is possible to obtain a Frege proof of $T$, assuming we already have a $P$-proof $\pi$ of $T$: we simply plug-in the encodings
$[\pi]$ and $[T]$ in the reflection principle, which makes the premise of the implication
true.
\medskip

Let $F$ be a non-commutative formula over $GF(2)$ and let $\overline{Q}^\phi(\overline{x})$ denote
the vector $(Q_1^{\phi},\ldots, Q_m^{\phi})$ (see Theorem \ref{thm:Frege_sim_ncIPS}).
Since $F$ is a \NCFIPS\ refutation of $\phi$ we know that
\begin{equation}\label{properties}
 F\left(\xZero\right)=0,~~~~~~~~~~ F\left(\xQx\right)=1\,.
\end{equation}
We can treat $F$ as a Boolean formula in the standard way:

\begin{definition}[$F_{\bool}$] \label{def:bool}
Let $F(\overline x)$ be a non-commutative formula over $GF(2)$ in the (algebraic) variables $\overline x$. We denote by $\Fb (\overline p)$ the Boolean formula in the (propositional) variables $\overline
p$ obtained by turning every plus gate and multiplication gate into $\oplus$   (i.e., XOR) and $\land$ (i.e., AND) gates, respectively, and, for the sake of clarity, turning the input algebraic variables $\overline x$ into the propositional variables $\overline p$. We sometimes write $F$ and $\Fb$ without explicitly mentioning the $\overline x$ and $\overline p$ variables. 
\end{definition}
Note that for any 0-1 assignment, $F$ and $\Fb$ take on the same  value (when we identify \textsf{true}
with 1 and \textsf{false} with 0).
When we consider $F=F(\overline x,\overline y)$ (with both the $\overline x$ and $\overline y$ variables),
$\Fb$ denotes the corresponding Boolean version of $F$  where the variables $\overline x$ are replaced by $\overline p$ and the algebraic variables $\overline y$ become
the propositional variables  $\overline y$.  Therefore, by \eqref{properties},
\begin{equation}\label{toProve}
   \Fz ~~~~~~\hbox{ and} ~~~~~~~~\Fq
\end{equation}
are both \textit{tautologies} (though we still need to show  that their Frege proofs are \textit{short}).
To conclude Theorem \ref{thm:Frege_sim_ncIPS}, we first prove $\neg\phi$ with a polynomial-size Frege
proof, assuming we already proved \eqref{toProve} (this
is done in Lemma \ref{soundness} below,  which is not very hard to establish). Second, we show that there exists an $s^{O(\log s)}$ Frege proof of  $\eqref{toProve}$ (which is done in Theorem \ref{thm:Fbool-in-quasipolynomial-Frege} in the next section, and
requires much more work).


\begin{lemma}\label{soundness}
There is a polynomial-size Frege proof of $\neg\phi(\overline p)$, assuming $\Fz$ and $ \Fq $ (polynomial in the size of $\phi(\overline p)$ and $\Fq$). 
\end{lemma}

\begin{proof}
By simple logical reasoning inside
Frege. Informally, we show that
 assuming that $\phi(\overline p)$ holds, for every $i\in[m]$, $\Qi\equiv 0$, and so $\Fz$ and $\Fq$ cannot both hold (for $\equiv$ denoting (semantic) logical equivalence).  

In order to go from $\phi(\overline
p)$ to  $\Qi\equiv 0$ we need to
deal with encoding of clauses
inside Frege. Thus, let  the propositional\ formula $\truthTable$ express  the statement that the assignment $\overline{p}$ satisfies the formula $\phi$, as defined below. In the following we  denote by  $\overline p$ \textit{actual} propositional variables occurring in a propositional formula (and not the  encoding of variables; see below).
\para{Encoding of 3CNFs and their truth predicate.} We shall follow Section 4.3 in \cite{GP14}.
A positive natural number $i$ is encoded with $\lceil\log_2 n\rceil$ bits (such that
the numbers $1,\ldots,2^t$ are put into bijective correspondence with $\{0,1\}^t$). We denote this encoding of $i$ by $[i]$. A clause $\kappa$  with three
literals is encoded as the bit string  $\overline q_1 s_1 \overline q_2 s_2 \overline q_3 s_3, $ where each $s_1,s_2,s_3$ is the sign bit of the corresponding literal in $\kappa$ (1 for positive and 0 for negative), and  each $\overline q_1,\overline q_2,\overline q_3$ is a length-$\lceil \log_2 n \rceil$ bit string encoding the corresponding index of the variable (assuming the number of variables is $n$). 
For a bit string $\overline q$ with the $i$th bit  $y_i$ we write $\overline q=[t]$
as an abbreviation of $\bigwedge _{i=1}^{\lceil \log_2 n\rceil} \left( y_i\sequiv [t]_i\right)$.  Finally, we define (where $m$ is the number of clauses in the 3CNF)
\begin{gather*}
        \truthTablek:=\bigvee_{j\in[3]}\bigvee_{i\in [n]} \left({[\overline q_j]=i}\wedge         (p_i\sequiv s_j)\right)\,,~~~\hbox{ and}
        \\
        \truthTable:=\bigwedge_{j\in[m]}\truthTablekj\,.
\end{gather*}
\begin{note}
It is important to note that, given a \textit{fixed} CNF $\phi$, the propositional formulas $\truthTablek$ and $\truthTable$ are formulas \textit{in the propositional variables $\overline p$ only.}   
\end{note}

\medskip

Let us now continue the proof of  Lemma \ref{soundness}. It is easy to show (see \cite{GP14} Lemma 4.9 for a proof) that after simplifying constants (e.g., $(A\land 1) \sequiv A$) the formula
$\truthTablek$ becomes syntactically identical to $\kappa(\overline p)$  and thus there
is a polynomial-size
Frege proof of 
\begin{equation}\label{phi2truth}
  \phi(\overline p)\rightarrow \truthTable.
\end{equation}

We shall now proceed within Frege. First, consider the propositional formulas 
\begin{equation}\label{eq:stam}
\truthTableki \rightarrow  \neg\Qi, \ ~~~~~\hbox{for all } i\in [m]\,.
\end{equation}

By definition (recall that $Q_i^{\phi}=0$ iff $\kappa_i$ is \textsf{true} for any 0-1 assignment, when 1 is identified with \textsf{true}) all the formulas in \eqref{eq:stam} are tautologies. Note that all the premises and all the consequences in \eqref{eq:stam} are of constant size, and thus \eqref{eq:stam} can be proved with a Frege proof of constant-size for each $i\in[m]$ (by completeness).
Further, using the fact that $\truthTable=\bigwedge_{i\in[m]}\truthTableki$, we can easily prove in Frege with a polynomial-size proof  that for each $i\in[m]$,
\begin{equation}\label{truth2Q}
      \truthTable \rightarrow \neg \Qi.
    \end{equation}
Assume by a way of contradiction
that $\phi(\overline p)$ holds. By modus ponens using \eqref{phi2truth} and \eqref{truth2Q}, we  have 
\begin{equation}\label{Qip}
  \bigwedge_{i\in[m]} \neg \Qi.
\end{equation}
We now argue (inside Frege) that, assuming also
$\Fz$, \eqref{Qip} implies $\neg\Fq$.
 By \eqref{Qip}, for every $i\in[m]$, $\Qi$ is logically equivalent to 0 (which is identified with \textsf{false}), and hence $\neg\Fq \equiv \Fz$. Thus, assuming $\Fz$ we have also $\neg\Fq$. But this contradicts
the assumption that $\Fq$, and
hence we reach a contradiction
with  the assumption
that $\phi(\overline p)$ holds.
\end{proof}
\bigskip

It remains to show a quasi-polynomial-size proof of \eqref{toProve}.
We abbreviate $\Fz$ and $\neg\left(1\oplus\Fq\right)$ by
\begin{equation}\label{finalStatement}
  \Fb'(\overline p),~~\Fb''(\overline p)\hbox{,~~ respectively}.
\end{equation}
Note that the substitutions of the constants $0$ or the constant depth formulas $\Qbar$  in $F$  cannot increase the depth of $F$ too much (i.e., can add at most a constant to the size of $F$). In other words, the depths of the formulas in \eqref{finalStatement}  are still $O(\log s)$.

\begin{proof}[Proof of Theorem \ref{thm:Frege_sim_ncIPS} (Second
main theorem)]
Using Theorem \ref{thm:Fbool-in-quasipolynomial-Frege} that we  prove below, we get that  \eqref{finalStatement} can be proved in quasi-polynomial-size (in $s$ the size
of the IPS refutation of the given CNF $\phi$).
And  together with Lemma \ref{soundness} above, this shows that $\neg \phi$ can be proved in quasi-polynomial-size in $s$, concluding the proof.
\end{proof}

\subsection{Non-commutative formula identities have quasi-polynomial-size proofs }

Recall that a (commutative or non-commutative) multivariate polynomial \(f\) is \emph{homogeneous} if every monomial in $f$ has the same total degree. For each $0\le j\le d$, denote by $f^{(j)}$ the homogenous part of degree $j$ of $f$, that is, the sum of all monomials (together with their coefficient from the field) in $f$ of total degree $j$. We say that a \textit{formula} is \emph{homogeneous} if each of its gates computes a \emph{homogeneous} polynomial (see Definition \ref{def:nonc_formula}
for the definition of a polynomial computed by a gate in a formula).
We shall use the following technical definitions:
\begin{definition}[Syntactic-degree]\label{def:syn_deg}
Define the \emph{syntactic degree} of  a non-commutative formula
$F$, $\deg (F),$ as follows:
(i) If $F$ is a field element or a variable, then $\deg (F)$ = $0$ and $\deg (F) = 1$, respectively;
(ii) $\deg(F+G) = \max(\deg (F), \deg(G))$, and $\deg(F\times G) = \deg (F) + \deg (G)$, where $+,\times$ denote  the plus and product gates respectively.
\end{definition}

\begin{definition}[Syntactic homogenous non-commutative formula] 
We say that a non-commutative formula is \emph{syntactic homogenous} if for every plus gate $F+G$ with two children $F$ and $G$, $\deg (F)=\deg(G)$.
\end{definition}

To complete the  proof of Theorem \ref{thm:Frege_sim_ncIPS} it 
remains to prove the following theorem:

\begin{theorem}
\label{thm:Fbool-in-quasipolynomial-Frege}
 If a non-commutative formula $F(\overline x)$  of size $s$ and depth $O(\log s)$ computes the identically zero polynomial over $GF(2)$, then the corresponding Boolean formula $\neg\Fb(\overline p)$ admits a Frege proof of  size  $s^{O(\log s)}$.
\end{theorem}

The rest of the paper is dedicated to proving this theorem.

%
%
%
%


\subsection{Remaining proof overview}
\label{sec:remain_proof_overview} 

For the convenience of the reader we highlight here in an informal manner the main steps in the proof of Theorem \ref{thm:Fbool-in-quasipolynomial-Frege} stating that the Boolean versions of balanced non-commutative formulas $F$ computing the zero polynomial over $GF(2)$
have Frege refutations (i.e., proofs of negation) of quasi-polynomial size.  
\begin{enumerate}

\item 
We prove \textit{inside Frege} that $\Fb$ can be partitioned into its Booleanized syntactic homogenous components $\Fb(i)$, each of quasi-polynomial size in $|\Fb|$, using Raz' \cite{Raz13-tensor} construction.
               
So it remains to show that each $\neg\Fb^{(i)}$ has a polynomial-size (in the size of $\Fb^{(i)}$, which is quasi-polynomial in the size of $F$) Frege proof.


\item
If $\Fb^{(i)}$ does not contain variables it is easy to refute $\Fb^{(i)}$. Note that
$\Fb^{(i)}$ does not contain variables iff $F^{(i)}$ does not contain variables. So we can assume that $F^{(i)}$
is non-constant. In this case we show that Frege can easily prove that $F^{(i)}$
is equivalent to some $\Fb'^{(i)}$, where $F'^{(i)}$ is a \emph{constant-free} (namely,
it does not contain constants) arithmetic non-commutative syntactic homogenous formula computing the zero polynomial.

\item 
Having a constant-free arithmetic non-commutative syntactic homogenous formula computing the zero polynomial $F'^{(i)}$ we use similar ideas as Raz and Shpilka \cite{RS04}
to construct a polynomial-size Frege refutation of (the Boolean version of) $F'^{(i)}$,
as follows:
\begin{enumerate}
\item 
\label{it:F_to_ABP}
Since $F'^{(i)}$ is constant-free and syntactic homogenous, using the standard
transformation \cite{RS04,Nis91} of a non-commutative formula to an algebraic branching
program (ABP;  Definition \ref{def:ABP}) results in a \textit{layered} (i.e., standard) ABP $A$ (this is different from  \cite{RS04} who had to deal with non-layered ABPs first). Assuming the syntactic degree of $F$ is $d$, the final layer
of $A$, consisting of the sink, is also $d$. 

\item 
\label{it:ABP_wit}
Using the ABP $A$, we identify a collection of witnesses that witness the fact that
$A$ computes the zero polynomial. Informally, these witnesses are  a collection of 0-1 matrices. Each 0-1 matrix denoted $\Lambda_i$ has a small number of rows (proportional
to the size of the ABP). Each (possibly zero) row $\mathbf v$ of these matrices corresponds to a  linear combination $\mathbf
v\cd \overline A_{d-i}$ (where $\cd$ is the inner product) of the polynomials computed  by the ABPs whose
sources are  in the $i$th layer of $A$ and whose sinks are all the (single) original sink of $A$ (and thus each of these ABPs computes a degree $d-i$ homogenous polynomial). The requirement is that for all such rows, $\mathbf v\cd \overline A_{d-i} = 0$,
and so overall $\LL_i \overline A_{d-i}=0$.
Note that each non-commutative polynomial in $\overline A_{d-i}$ is homogenous of degree
$d-i$.

Following a similar argument to \cite{RS04}, we show that one can find matrices $\Lambda_i$
such that, in addition to the above requirement, the following holds: 
$$
\Lambda_i \overline A_{d-i} = \TT_{i+1}\Lambda_{i+1}\overline A_{d-i-1},
$$
where $\TT_1\.\TT_d$ are matrices with homogenous linear forms in every entry, and
such that the product of the matrix $\TT_{i+1}\LL_{i+1}$ with $\overline A_{d-i-1}$
is construed  in a \textit{syntactic} way; that is, $\TT_{i+1}\LL_{i+1}$ is interpreted
as an adjacency matrices of a layer in an ABP where the $(l,k)$ entry of the matrix is the linear form that labels the edge going from the
$l$th node in layer $i$ to the $k$th node in layer $(i+1)$---and so $\TT_{i+1}\Lambda_{i+1}\overline A_{d-i-1}$ is a new ABP with $d-i$ layers.

\item
Note that it is unclear how to (usefully) represent an ABP directly in a Frege system, because apparently ABP is a stronger model than formulas (and each Frege proof-line is written as a formula). Thus, we cannot directly work with ABPs within Frege proofs, and consequently
we cannot use the witnesses from part (\ref{it:ABP_wit}). We solve this problem by
replacing every ABP in the witnesses by a corresponding non-commutative formula: every
ABP in the witnesses from (\ref{it:ABP_wit}) is a part of the ABP that was constructed
from  $F'^{(i)}$  in (\ref{it:F_to_ABP}). We notice that every such part of ABP corresponds
to a certain substitution instance of $F'^{(i)}$. Thus, we replace every such part of ABP in
the witnesses with its corresponding substitution instance of $F'^{(i)}$. Having these
witnesses enables
us to carry out a step by step proof of the fact that $F'^{(i)}$ computes the zero polynomial (formally,
a Frege refutation  of the Boolean version of $F'^{(i)}$).      

\end{enumerate} 

\end{enumerate}

\begin{proof}[Proof of Theorem \ref{thm:Fbool-in-quasipolynomial-Frege}]
The formula $F$ is of size $s$ which means that the maximal degree of a polynomial computed by $F$ is at most $s+1$. Raz  $\cite{Raz13-tensor}$ showed that we can always split $F$ into syntactic homogenous formulas  $\hF{i},\;i=0\.s+1$, each of size  $s^{O(\log s)}$. In Lemma \ref{homogenous-proof}, proved in the next section, we show that this homogenization construction can already be proved efficiently in Frege. In other words,  we show that there exists an $s^{O(\log s)}$-size Frege proof of

\begin{equation}\label{homogenous-equality}
  \bigoplus_{i=0}^{s+1}\hFbool{i}\leftrightarrow \Fb.
\end{equation}

By Theorem \ref{thm:homo-formula} proved in the sequel, for any \emph{syntactic homogenous} non-commutative formula $H$ that computes the identically zero polynomial over $GF(2)$,   $\neg H_{bool}$ admits a polynomial-size (in the size
of $H$) Frege proof (recall that $\neg H_{\bool}$ is a tautology whenever $H$ is a non-commutative formula computing the zero polynomial over $GF(2)$). Thus, by Theorem \ref{thm:homo-formula}, for every  $\hF{i}$, $i=0,\ldots,s+1$, there exists an $s^{O(\log s)}$-size Frege proof of $\neg\hFbool{i}$. That is, there exists an $s^{O(\log s)}$-size Frege proof of $\neg \left(\bigoplus_{i=0}^{s+1}\hFbool{i}\right)$. Note that Theorem \ref{thm:homo-formula} gives proofs that  have size
polynomial in the size of $\neg\hFbool{i}$, and this latter size is $s^{O(\log s)}$.
Together with  tautology \eqref{homogenous-equality}, we can  derive
 $\neg\Fb$ in Frege.
\end{proof}

\subsection{Proving the homogenization of non-commutative formulas in Frege }

To complete the proof of Theorem  \ref{thm:Fbool-in-quasipolynomial-Frege}
it remains to prove Lemmas \ref{homogenous-proof} and \ref{thm:homo-formula}.
Lemma \ref{homogenous-proof} states that Raz' construction from \cite{Raz13-tensor} for homogenizing
arithmetic formulas  is efficiently provable in Frege (and is also applicable to non-commutative formulas):
\begin{lemma}\label{homogenous-proof}
If $F$ is a non-commutative  formula of size $s$ and depth $O(\log s)$  and $\hF{0}\.\hF{s+1}$ are the syntactic homogenous formulas computing $F$'s homogenous parts of degrees $0\.s+1$, respectively,
 constructed according to \cite{Raz13-tensor} (sketched below), then there exists an $s^{O(\log s)}$-size Frege proof of:
\begin{equation}\label{equ:main-homogenous-equ}
  \left(\bigoplus_{i=0}^{s+1}\hF{i}\right)\leftrightarrow \Fb.
\end{equation}
\end{lemma}
\begin{proof}
We first introduce basic notations and observations for describing Raz' (commutative) formula homogenization construction from \cite{Raz13-tensor}. This construction is a somewhat more involved variant
of the standard  homogenization construction for circuits laid out by Strassen \cite{Str73}. We then construct the desired short Frege proofs, which in turn also shows how to construct
the homogenous  formulas themselves.  

\para{Raz' formula homogenization construction.} 

Given a balanced (commutative) arithmetic formula $F$ we wish to construct  $s$ (commutative)
formulas computing the homogenous parts $\hF{i},\; i=0\.s$. Define the \emph{product-depth} of a gate $u$, denoted $\pdu$, as the maximal number of product gates along a directed path  from $u$ to the output gate  (including $u$). Since the formula $F$ is balanced, the  depth is  at most $O(\log s)$, namely the largest value of $\pdu$ for any node $u$ in $F$ is $O(\log s)$. 

Let us consider the directed path from $u$ to the root (including the node $u$). Informally we want to describe a possible progression of the degree of a  monomial computed along the path from $u$ to the root. Observe that any possible degree progression must occur on product gates. That is, for a gate $u$ with product-depth $\pdu$,  there are $\pdu$ ``choices'' for the degree of a monomial to increase (i.e., ``to progress'').

Formally, for every integer $r$, denote by $N_r$ the family of monotone non-increasing functions $D$ from $\set{0,1 ,\ldots, r}$ to $\set{0,1 ,\ldots, s+1}$. It is helpful to think of $D$ as a function from product nodes along the directed path to the corresponding degree of a monomial as it is computed along the path, where the path starts from the root (identified with $0$) and terminates with the product gate closest to $u$ (including possibly $u$ itself; identified with $r$). Thus, for instance, the root $0$ is mapped
to the total degree of the monomial. Therefore, the set $\Gd$ describes all possible progressions of the degree of  monomials along the path from $u$ to the root.
Note that a product gate may not increase the degree of a monomial computed along a path, because we may consider the monomial as multiplied by  a constant. Hence, the functions in $N_r$ are not necessarily strictly decreasing. 

The size of $N_r$ is ${r+s+2 \choose  r+1} = {r+s+2 \choose s+1}$ (the number of combinations with repetitions of
       $r+1$ elements from $s+2$ elements, which  determine  functions in $N_r$). Therefore, for every node $u$ in $F$, the size of the set $\Gd$ is at most 
\begin{equation}\label{eq:prod_depth}
{s+O(\log s)+2 \choose s}= s^{O(\log s)}.
\end{equation}

 We construct the desired syntactic homogenous formulas $\hF{0},\hF{1}\.\hF{s+1}$ by constructing a  formula $\bF$ according to $F$. Split every gate $u$ in $F$ into $|\Gd|$ gates in $\bF$, labeled $(u,\x)$, for every $\x \in \Gd$. We will add edges connecting nodes in $\bF$ the same way as \cite{Raz13-tensor}.
It might be helpful for the reader to consult \cite{Raz13-tensor} to get the intuition of the construction itself, however our presentation   is self-contained, as we will show how the construction is efficiently provable already inside the Frege system (this will also show that $\bF$ and $F$ compute the same polynomial, when considered as  Boolean functions).

%
%

Denote by $\bF_{u,\x}$ the subformula rooted at $(u,\x)$ in $\bF$. There might be
some isolated nodes $(u,\x)$, namely nodes that no edge connects to them, and we consider the subformulas on these nodes as $0$. Similarly, denote by $F_u$ the subformula rooted at $u$ in $F$.
In \cite{Raz13-tensor} it was demonstrated (see below) how to construct $\bF$ so that for every node $(u,\x)$ in $\bF$, $\bF_{u,\x}$ is a homogenous formula computing the degree-$\x(\pdu)$ homogenous part of $F_u$.
More precisely, for  every node $u$ in $F$, denote by $s_u$  the size of the formula $F_u$. The maximal degree of the polynomial computed by $F_u$ is $s_u+1$. For $i=0\. s_u+1$, let $\Du{i}$  denote the \textit{set} of all functions $\x$ in $N_{\pdu}$ such that $\x(\pdu)=i$.  For \emph{any two}  $D,D'\in\Du{i}$, the formulas $\bF_{u,D}$
and $\bF_{u,D'}$  are \textit{identical} (and compute the homogenous part of degree $i$ of $F_u$). Thus we can consider $\bF_{u,\Du{i}}$ as a \textit{single} formula.  



\para{Efficient proofs of the homogenization construction.}

\newcommand{\su}{{s_u}}


Next, we use a similar inductive argument as in \cite{Raz13-tensor}, from leaves to the top gate of $F$, showing that for every gate $u$  in $F$ there exists an $s^{O(\log s)}$-size Frege proof of 
\begin{equation}
\label{eq:homo-proof-for-each-gate}
\left(
        \bigoplus_{i=0}^{s_u+1}\bF_{u,\Du{i}\; bool}
\right)
        \leftrightarrow F_{ u\; bool}\,.
\end{equation}
Observe that the formulas $\bF_{r,\D{0}{r}}\.\bF_{r,\D{s+1}{r}}$, for $r$ being the root of $F$, are just those desired formulas $F^{(0)},F^{(1)}\.F^{(s+1)}$.
Thus, eventually, when we prove \eqref{eq:homo-proof-for-each-gate}
for the root node $r$, we prove the  existence of an $s^{O(\log s)}$-size Frege proof of the Boolean formulas in \eqref{equ:main-homogenous-equ}.

Note that the size of the Frege proof we construct is  quasi-polynomial in $s$.  This is because
of the following: for
every node $u$ in $F$ and every $\x \in \Gd$ we construct a proof of \eqref{eq:homo-proof-for-each-gate}.
Recall that  for every $u$ in $F$, $|\Gd|= s^{O(\log s)}$, by \eqref{eq:prod_depth}.
Thus, the total number of such nodes $u$ and functions $\x$ is 
$s\cd  s^{O(\log s)} =  s^{O(\log s)}\,,$
and so this is the total number of proofs of \eqref{eq:homo-proof-for-each-gate} we
construct. 
Each such proof of \eqref{eq:homo-proof-for-each-gate} requires only $\poly(s)$
size \emph{assuming} we already have proved (by induction hypothesis) the required
previous instantiations of \eqref{eq:homo-proof-for-each-gate}. Therefore, we end up
with a proof of total size $\,\poly(s)\cd s^{O(\log s)}= s^{O(\log s)}$. 

\renewcommand{\v}{D^v}
\newcommand{\w}{D^w}

\Base If $u$ is a leaf,  for each $\x\in N_{\pdu}$,  the node $(u,\x)$ is defined to
be a leaf of $\bF$. Furthermore, if $u$ is labeled by a field element, $(u,\x)$ is labeled by the same field element in case $\x(\pdu) = 0$ and by $0$ in case  $\x(\pdu)\neq 0$. If $u$ is labeled by an input variable, $(u,\x)$ is labeled by the same input variable in case $\x(\pdu) = 1$ and by $0$ in case $\x(\pdu) \neq 1$. Thus, for each $\x\in N_{\pdu}$, either $\bF_{(u,\x)}$ computes $F_u\h{\x(\pdu)}$ or $0$. Namely, we can easily prove in Frege
$$
\left(
        \bigoplus_{i=0}^{\su+1}\bF_{u,\Du{i}\; bool}
\right)
        \leftrightarrow F_{u\; bool}.
$$
                               
\induction   

\case 1
Assume that  $u$ is a sum gate with children $v,w$. For every $\x\in \Gd$, let  $\v\in N_{\pd{v}}$ be the function that agrees with $\x$ on $\set{0,1\.\pdu}$ and satisfies $\v(\pd{v}) = \x(\pdu)$, and in the same way, let $\w\in N_{\pd{w}}$ be the function that agrees with $\x$ on $\set{0,1\.\pdu}$ and satisfies $\w(\pd{w}) = \x(\pdu)$. The node $(u,\x)$ is defined as 
$$
\bF_{u,\x}: = \bF_{v,\v}+\bF_{w,\w}.
$$
Assume $\x(\pdu)=j$. Then it means 
$$
\hat\bF_{u,\Du{j}}: = \hat{\bF}_{v,\Dv{j}}+\hat{\bF}_{w,\Dw{j}}
$$
(recall that given a non-commutative formula $F$,
$\hat F$ denotes the non-commutative  \textit{polynomial} it computes). Therefore, the following is a tautology:
$$
\bF_{ u,\Du{j}\;bool}\leftrightarrow
\left(
\bF_{v,\Dv{j}\;bool}\oplus\bF_{w,\Dw{j}\;bool}
\right),~\hbox{
for all } j=0\.s+1.
$$
\renewcommand{\su}{{s_u+1}} By induction hypothesis on the nodes $v,w$, we have
$$
F_{v\; bool}\leftrightarrow \bigoplus_{i=0}^{s_v+1}\bF_{v,\Dv{i} \;bool},~~~~~ ~~F_{w\; bool}\leftrightarrow \bigoplus_{i=0}^{s_w+1}\bF_{w,\Dw{i}\;bool},$$
                                and so we can prove
                                $$\bigoplus_{i=0}^\su\left( \bF_{v,\Dv{i}\;bool}\oplus \bF_{w,\Dw{i}\;bool}\right)  \leftrightarrow F_{ v\;bool}\oplus F_{w\;bool},
$$
which gives us (since $u$ is a plus gate)
$$
\bigoplus_{i=0}^\su\bF_{u,\Du{i}\;bool}\leftrightarrow F_{u\;bool}\,.
$$
\smallskip                               

\case 2 
If $u$ is a product gate with children $v,w$, using the same notation as above, for $j=0\.s_u+1$, we define $\bF_{u,\Du{j}}:=\sum_{i=0}^j \bF_{v,\Dv{i}}\cdot \bF_{w,\Dw{j-i}}.$ If $j>u_s$, let $\bF_{u,\Du{j}}:=0$.
                                Similarly, using the induction hypothesis on the nodes $v,w$ and observing the fact that  $s_u=s_v+s_w+1$, we can prove:
                                $$F_{u\;bool}\lequal \bigoplus_{i=0}^\su\bF_{u,\Du{i}\; bool}$$ as follows:
                                \begin{align*}
                                    F_{u\;bool} \lequal& \left(F_{v\;bool}\land F_{w\;bool}\right)\\
                                        \lequal&\left(\bigoplus_{j=0}^{s_v+1}\bF_{v,\Dv{j}\;bool}\right)\land\left(\bigoplus_{i=0}^{s_w+1}\bF_{w,\Dw{i}\;bool}\right)\\
                                        \lequal& \bigoplus_{j=0}^{s_v+s_w+2}\bigoplus_{i=0}^{j}\left(\bF_{v,\Dv{i}\;bool}\land\bF_{w,\Dw{j-i}\;bool}\right)\\
                                        \lequal&\bigoplus_{j=0}^\su\left(\bigoplus_{i=0}^{j}\left(\bF_{v,\Dv{i}\;bool}\land\bF_{w,\Dw{j-i}\;bool}\right)\right)\\
                                        \lequal &\bigoplus_{i=0}^\su\bF_{u,\Du{i}\;bool}\,.
                                \end{align*}
\end{proof}

\subsection{Homogenous non-commutative formula identities have polynomial-size
Frege proofs}\label{sec:homog_idens_have_psize_proofs}

To conclude Theorem \ref{thm:Fbool-in-quasipolynomial-Frege} it remains to prove
Theorem \ref{thm:homo-formula}. Here we will prove Theorem \ref{thm:homo-formula},
based on further lemmas we prove in the next section. 

First, we need to set some notation. We denote by 
$$
F\synEqual F'
$$ 
the fact that  $F'$  can be derived
with a polynomial in $|F'|$ size Frege proof, given $F$ as a (possibly empty) assumption in the proof. 

In practice we will almost always use this notation  when $F'$ can be derived form $F$ 
by  simple syntactic manipulations of formulas using mostly structural rules such as the associativity and distributivity rules, as well as simple logical identities (e.g., ${\sf false} \oplus G \equiv G$, where $\oplus$ stands for  XOR). This notation will make our arguments a bit more convenient to read. For most part, we will also leave it to the reader to verify that  indeed $F'$ can be obtained from $F$ with a short Frege proof, since it will be evident from the way $F$ and $F'$ are defined. 

Accordingly, for two \textit{vectors} of formulas $\overline F,\overline G$, we
denote by $\overline F\synEqual \overline G$ the fact that each entry 
in $\overline G$ can be derived from the corresponding entry in $\overline F$ with a short Frege proof.



The following definition is essential to Section \ref{sec:sub-sec-on-implicit-ABP} where
we talk about algebraic branching programs  (ABPs). This definition
will enable us to identify
within a homogenous non-commutative formula a certain part of the formula (after substitution) that corresponds to a sub-algebraic branching program. 
\begin{definition}[Induced part of a formula]\label{def:part-of-formula}
Let $F'$ be a subformula of F and $g_1\.g_k$ be gates in $F'$ and $c_1\.c_k$ be  constants in $\F$. Then $F'[c_1/g_1\.c_k/g_k]$ is called \emph{an induced part of $F$}.
\end{definition}
We sometimes call an induced part of a formula simply a \emph{part of a formula}.

\subsubsection*{Getting rid of constants}
For technical reasons (concerning the conversion of a non-commutative syntactic homogenous formulas into a layered  ABP in what follows) it will be convenient to consider only arithmetic (resp.~Boolean) formulas with \textit{no}  0-1 (resp.~\textsf{true, false}) constants. We say that a Boolean or arithmetic formula is \textit{non-constant} if it contains at least one variable. 

\begin{lemma}[Constant-free formulas]
\label{lem:discard_constants}
Let $F$ be a non-constant and non-commutative formula over $GF(2)$ that computes the (non-commutative) zero-polynomial. Then, there exists a \emph{constant-free} non-commutative formula $F'$ of size ${\rm poly}(|F|)$ that computes the (non-commutative) zero polynomial, such that $\Fb\synEqual
\Fb'$.
\end{lemma}

Note that since $F$ does not contain 0-1 constants, $F'$ does not contain \textsf{true, false} constants.

\begin{proof}
First, notice that substitution of equivalent terms can be simulated efficiently in Frege, in the following sense: if $\Phi$ is a formula and $\psi$ is a subformula occurring in $\Phi$, then $\psi\sequiv\psi'\synEqual
\Phi'$, where $\Phi'$ is $\Phi$ in which the (single) occurrence $\psi$ is substituted by $\psi'$.

Therefore, we can iteratively take the constants out of $\Fb$ within Frege using local substitution of logically equivalent terms, as follows: if $F$ contains a subformula $0+G$, for some formula $G$, we change it to $G$; if $F$ contains a subformula $1+G$ we change it to $\neg G$; if $F$ contains a subformula $0\times G$, we change it to $0$; if $F$ contains a subformula $1\times G$ we change it to $G$. Doing these replacements iteratively we arrive at either the $0$ formula or a formula without constants (since every step reduces the size of the formula). The $0$ formula is arrived only when there are no  variables in $F$, and so this cannot happen by our assumption. 
\end{proof}

\subsubsection*{Main technical theorem}

\begin{theorem}[restatement of Theorem \ref{thm:intro:homo-formula}] \label{thm:homo-formula}
There exists a constant $c$ such that for any non-commutative syntactic homogeneous
formula $F(\overline x)$  over $GF(2)$ of size $s$ that is identically zero, the corresponding Boolean tautology  $\neg\Fb(\overline p)$ has a Frege proof of size at most $s^{c}$ (for sufficiently large $s$).
\end{theorem}
\begin{proof}


First, by Lemma \ref{lem:discard_constants} we can assume without loss of generality
that  $F(\overline x)$ is constant-free (or else, we can either derive an equivalent
constant-free formula or simply the constant \textsf{false}, both with  polynomial-size
Frege proofs). 
Thus, assume from now that  $F$ is constant-free and let $d$ be the syntactic-degree of $F$. 

Note that the syntactic-degree $d$ of $F$ is at most $s+1$.
Theorem \ref{existLambda}, proved in the  next section, states the
existence of a collection of witnesses that  witness that the homogenous
non-commutative constant-free formula    $F$ computes the non-commutative
zero polynomial.  As demonstrated below, these witnesses  will enable us to inductively and efficiently prove in Frege that $F$ is the zero polynomial (over $GF(2)$). 

First, we give the formal description of the witnesses and their properties and then explain informally why they
witness the identity and why they exist for every identity.

\para{Notation.}
For a matrix $\TT$ with entries $\TT_{ij}$, each a non-commutative formula,  and a vector $\FF=(v_1,\. v_m)$ of non-commutative formulas, we write $\TT\FF$ to denote the (transposed) vector of non-commutative  \textit{formulas }whose $j$th entry is $\TT_{j1}\times F_1+\ldots+\TT_{jn}\times F_m$ written as a  balanced (depth $\le \log m+1$) binary tree of plus gates at the top and the formulas $\TT_{jk}\times
F_k$'s at the leaves. For a $0$-$1$ matrix $\LL$, we write $\LL \FF$ to denote the vector of non-commutative formulas similar as defined above for $\TT \FF$, except that now the matrix $\TT$ has the a 0-1 formula
in each entry.
We denote by $\TT\LL\FF $ the vector of non-commutative formulas $(\TT\LL)\FF$, where the $(i,j)$ entry of the
matrix $\TT\LL$ is $ \sum_{k}\TT_{ik}\LL_{kj}$ written as a balanced tree of plus gates with corresponding
leaves as before (and where $\TT_{ik}\LL_{jk}$ is written as $\TT_{ik}$ if $\LL_{kj}=1$ and does not occur in the sum if $\LL_{kj}=0$).
When we write $\synEqual A \sequiv B$ in the witnesses below, for $A$ and $B$ non-commutative
arithmetic formulas over $GF(2)$, we intend to treat $A\sequiv B$ as a Boolean tautology
(Definition \ref{def:bool}). For two vectors of formulas $\mathbf v=(v_1,\ldots,v_m), \mathbf u=(u_1,\ldots,u_m)$, we write $\synEqual \mathbf v \sequiv \mathbf u$ to denote $\synEqual v_i\sequiv u_i$,
for all $i\in[m]$. \bigskip

\begin{minipage}[c]{0.85\textwidth}
\para{Identity Witnesses}

\begin{enumerate}

 \item For every $i=0\. d-1$, $\zeroMatrix$ is a 0-1 matrix of dimension $m_i\times m_i$,
where $m_i=\poly(|F|)$, for all $i$. We set $\LL_0=1$. 

  \item For every $i=1 ,\ldots, d-1$, $\TT_i$ is an $m_{i-1}\times m_{i}$ matrix whose entries are homogenous linear forms in the  $\overline x$ variables with 0-1 coefficients. 
\item For every $i=1\.d$,   $\homoji$ is a vector of  induced parts of $F$, each computing a homogenous non-commutative polynomial of degree exactly $i$. The length of the vectors $\homoji$ is  $m_{d-i}$.
Accordingly, we denote by \hhat\ the vector of non-commutative \emph{polynomials}
in $\homoji$.  
\end{enumerate}
These witnesses are such that the following hold:
\begin{equation}\label{intermidate1}
  \LL_{d-i} \hhat =\overline 0,~~~~i=1\. d\, \hbox{~~is a true equality;}\footnotemark
\end{equation} 
\begin{equation}\label{conclusion-Finalproof}
  \synEqual F \sequiv  \LL_0\homojiFinal \hbox{ (meaning that $\synEqual F \sequiv F_d$,
since $\LL_0=1$);}\footnotemark 
\end{equation}
\begin{equation}\label{intermidate2}
 \synEqual \LL_{d-i}\homoji  \sequiv \TT_{d-i+1}\LL_{d-i+1}\FF_{i-1}\,,~~~~~~i=2\.d\,.
\end{equation}
 \end{minipage}
\bigskip 
\addtocounter{footnote}{-1} 
\footnotetext{This is a \emph{semantic} equality.
I.e., in itself it does not entail a small proof of the equality.}
\addtocounter{footnote}{1}
\footnotetext{Note that $\homojiFinal$ is identical to $F_d$, because the only induced part
of $F$ of degree $d$ is $F$ itself, due to syntactic homogeneity.}

\para{Using the witnesses.} The identity witnesses provide a way to prove inductively
that the non-commutative syntactic homogenous and constant free formula $F$ is identically zero (when considered as
a Boolean formula over $GF(2)$). Informally, we start with
$\LL_{d-1}\FF_1=0 $ (considered as a Boolean equality) which is a true identity by \eqref{intermidate1}.
Since this identity is written as a sum of linear forms it has a polynomial-size proof. From this we also get    $ \TT_{d-1}\LL_{d-1}\FF_{1}=0$, and since by \eqref{intermidate2},
$\LL_{d-2}\FF_{2} = \TT_{d-1}\LL_{d-1}\FF_{1}$, we derive $\LL_{d-2} \FF_2=0$.
Continuing in this fashion we
finally derive  $\LL_0\FF_d  = 0$, which by \eqref{conclusion-Finalproof} concludes the proof.  

It is worth noting that we cannot directly represent the formula $F$ as the iterated matrix product $\TT_{d-1}\cdots \TT_2\LL_{1}\FF_{1}$ in the proof, since writing explicitly this iterated matrix
product will incur an
exponential-size blow-up. 
\medskip

In what follows we make the above argument  \textit{formal}. We demonstrate a proof of  $\neg \Fb$ based on the tautological Boolean formula obtained from  equation \eqref{intermidate1}
and the short Frege proofs of the tautological Boolean formulas in 
 \eqref{conclusion-Finalproof}
and \eqref{intermidate2}.
Denote by $\homojiBool$ the vector of all corresponding Boolean formulas of the formulas in $\homoji$, and let $\homojiBoolt$
be the $t$th coordinate of this vector.

By \eqref{intermidate1}, the following Boolean formulas are all
tautologies:
\begin{equation}  
\label{Raz-trick}
        \bigwedge_{w\in [m_{d-i}]}
                \left(
                        \neg\left(
                                \bigoplus_{t: \LL_{d-i}(w,t)=1}  
                                \homojiBoolt
                            \right)
                \right),~~~~i=1\. d.
\end{equation}

By \eqref{intermidate2}, for every  $i=2\. d$, and every $u\in[m_{d-i}]$, we have short Frege proofs of the following logical equivalence (between two Boolean formulas; the left hand side being the $u$th row in $\LL_{d-i}\homoji$):
\begin{multline}\label{equivalence}
        \left(\bigoplus_{t: \LL_{d-i}(u,t)=1} 
                 \homojiBoolt
        \right) 
                \sequiv \\
                \bigoplus_{w\in [m_{d-i+1}]}
        \left(\TT_{d-i+1}(u,w)_{bool}(\overline p)
                ~\land  
                 \left(\bigoplus_{t: \LL_{d-i+1}(w,t)=1}   
                      \homojiBoolFormert
                 \right)
        \right).
\end{multline}

\begin{claim*}\label{cla:some-additional-claim}
There are polynomial-size Frege proofs of \eqref{Raz-trick}.
\end{claim*}
This will conclude the proof of Theorem \ref{thm:homo-formula}, since for $i=d$, we get a polynomial-size  Frege proof of  
$$
        \bigwedge_{w\in [m_{d}]}
                \left(
                        \neg\left(
                                \bigoplus_{t: \zeroMatrixFinal(w,t)=1}   
                                \homojiBooltFinal
                            \right)
                \right),$$
which is just $\neg \Fb(\overline p)$ by \eqref{conclusion-Finalproof} ($m_d=1$ and
$\LL_d=1$).

\begin{proofclaim}
First fix $i=1$. Since  $\overline F_1$ is a vector of linear forms,  \eqref{Raz-trick} is
a Boolean tautology which can be proved with a polynomial-size Frege proof.
This means that when we fix $i=2$, the right hand side of \eqref{equivalence} becomes false for every $u\in[m_{d-2}]$, and so the left hand side of \eqref{equivalence}
$$
\bigoplus_{t: \LL_{d-2}(u,t)=1} \homojBool_{2,t}
$$
is also false for every $u\in[m_{d-2}]$. We continue in this manner until we arrive to \eqref{Raz-trick} for $i=d$.  
\end{proofclaim}

We have thus concluded the proof of Theorem \ref{thm:homo-formula}.
\end{proof}

%
%
%

\subsection{Identity witnessing theorem}\label{sec:splitting}

\newcommand{\q}{q}
\newcommand{\vlevel}{v_{i,1} ,\ldots, v_{i,m_i}}
\newcommand{\vtwoi}{v_{\q,t}}
\newcommand{\vthreei}{v_{\q+1,j}}
\newcommand{\vectorABPzero}{\overline {A_{0 }} }
\newcommand{\vectorABP}{\overline {A_{\q }}}
\newcommand{\vectorABPtwo}{\overline {A_{\q }}}
\newcommand{\vectorABPthree}{\overline {A_{\q+1 }}}
\newcommand{\vectorABPl}{\overline {A_{l }}}
\newcommand{\vectorABPlminus}{\overline {A_{l-1 }}}

\newcommand{\vectorABPLayer}{\ensuremath{{\overline {A_{i+1 }}}}}
\renewcommand{\AA}{\ensuremath{{\overline {A}}}}
\newcommand{\vectorABPLayerFormer}{\ensuremath{\overline {A_{i}}}}
\newcommand{\vectorABPLayerFinal}{\ensuremath{\overline {A_{l}}}}

\newcommand{\basisLayeri}{\overline \lambda_{i,1} \. \overline\lambda_{i,r_i}}
\newcommand{\basistwo}{\overline \lambda_{\q,1} \. \overline\lambda_{\q,r_\q}}
\newcommand{\basistwot}{\overline \lambda_{\q,t}}
\newcommand{\basistwoti}{\lambda_{\q,t,i}}

\newcommand{\basisthree}{\overline \lambda_{\q+1,1} \. \overline\lambda_{\q+1,r_{\q+1}}}
\newcommand{\basisthreet}{\overline \lambda_{\q+1,t}}

\newcommand{\basislt}{\overline \lambda_{l,t}}
\newcommand{\basisltminus}{\overline \lambda_{l-1,t}}
It remains to prove the following:

\begin{theorem}[Identity witnessing theorem]\label{existLambda}
 Let $F(\overline x)$ be a  non-commutative syntactic homogenous constant-free formula of degree $d$  over $GF(2)$ computing the non-commutative
zero polynomial. Then, the identity witnesses as defined in Section
\ref{sec:homog_idens_have_psize_proofs} exist.
\end{theorem}

The proof of this theorem uses the notion
of an algebraic branching program mentioned before as well as  the Raz and Shpilka PIT algorithm
\cite{RS04}. Our proofs are self-contained, and we demonstrate formally the existence of the
witnesses from scratch. 

\subsubsection{Algebraic branching programs}\label{sec:implicit_ABP}


We introduce the following definition:
\begin{definition}[ABP]\label{def:ABP} An \emph{algebraic branching program} (ABP
for short) is a directed acyclic
graph with one source and one sink. The vertices of the graph are partitioned into
\emph{layers} numbered from $0$ to $d$ (the \emph{degree} of the ABP), and edges may go only from
layer $i$ to layer $i + 1$. The source is the only vertex at layer $0$, and the sink is the
only vertex at layer $d$. Each edge is labeled with a homogeneous linear polynomial in
the variables $x_i$ (i.e., a function of the form $\sum_{i}c_ix_i$, with coefficients $c_i\in\F$, where \F\ is the underlying field). The \emph{size} of an ABP is the number of its vertices. A path, directed from source to sink, in the ABP is said to \emph{compute} the non-commutative product of linear forms on its edges (in the
order they appear on the path). A node in the ABP \emph{computes} the sum of all incoming paths arriving from the source. The ABP \emph{computes} the non-commutative
polynomial computed at its sink. 
\end{definition}

Note that by definition an ABP computes a \emph{homogenous} non-commutative polynomial. 

Raz and Shpilka \cite{RS04} established a deterministic polynomial-time algorithm for the polynomial identity testing of (non-commutative) ABPs. Therefore, by transforming a non-commutative formula to an ABP, one obtains a deterministic polynomial-time algorithm for the polynomial identity testing of
non-commutative formulas.
\begin{theorem}[Theorem 4, \cite{RS04}]\label{Raz_Shpilka_thm}
Let $A$ be an ABP of size $s$ with $d + 1$ layers, then we can verify whether A computes
the non-commutative  zero polynomial  in time
$O(s^5 + s\cdot n^4)$.
  \end{theorem}

\newcommand{\allNodesOnLayer}{A(v_{l-i,1},\vsink),\ldots, A(v_{l-i,m_{l-i}},\vsink)}
Using the algorithm demonstrated in Theorem \ref{Raz_Shpilka_thm}, we give in Lemma \ref{lem:Raz-ABP} below witnesses that certify that a given non-commutative formula computes
the zero polynomial. These witnesses \emph{will not be our final witnesses} because they will incorporate ABPs, whereas in Section \ref{sec:homog_idens_have_psize_proofs} we required the witnesses to consist of non-commutative \emph{formulas} and not ABPs. In the next section we show, based on Lemma \ref{lem:Raz-ABP}, how to obtain the desired formula-based
witnesses.

\para{Notation.}
For what follows in this section, let  $A$ be an ABP with $l+1$ layers and where the source node \vin\ is on the $0$th layer and the sink node \vsink\ is on the $l$th layer. For every $j=0\.l$, we denote the nodes on the $j$th layer by $v_{j1},\.v_{jm_{j}}$, where $m_j$ stands for the total number of nodes in the $j$th layer. For a given $i=0,\ldots,l$, consider the ABP \textit{with $m_{l-i}$ sources} in layer  $l-i$  and whose sink is \vsink. We can denote this multi-source ABP as a vector of ABPs:    
$$
\AA_i=\left(\allNodesOnLayer\right).
$$
Each entry in this vector computes a non-commutative homogenous degree $i$ polynomial. It is \textit{important to note} that $\AA_i$ is only a convenient notation, namely,
when we apply in what follows a matrix product to the vector $\AA_i$, we will treat different coordinates in the vector $\AA_i$ as having \emph{joint nodes}. For instance, the sink node $\vsink$ is treated as a \emph{single node} shared by
all the coordinates (and so in the vector $\AA_i$ it occurs only once). We will thus build
a \emph{single} ABP out of a matrix product with the vector $\AA_i$, as described in what follows.

For a 0-1 matrix $\LL$ of dimension $m\times m$ and a  \emph{multi-source} ABP $\overline A$ with $m$ sources
$v_1,\ldots,v_m$ and $0$ to $l$ layers, we write $\LL\overline A$ to denote the $l+1$ layered ABP with $m$ sources that results from $\overline A$ when we join together several sources into a single source, maintaining the outgoing
edges of the joined sources. Specifically, the $i$th source of the new ABP computes the non-commutative
polynomial $\sum_k^m\LL_{ik} v_k$, and this is done by defining the outgoing edges of the  $i$th source to be all the outgoing edges of $v_k$, for all $v_k$ such that $\LL_{ik}=1$.  

For a matrix $\TT$ with dimension  $m\times m'$ and  entries $\TT_{ij}$
that are homogenous linear forms, and a multi-source ABP $\AA=(v_1,\. v_m')$ we write $\TT\AA$ to denote the ABP whose $0$ layer
consists of $m$ sources, and the $i$th node in the 0th layer, for $i=1,\ldots,m$,  is connected to the $j$th node in 1st layer, for $j=1,\ldots,m'$, with an edge labeled by the
linear form $\TT_{ij}$.
In case $\LL$ is a 0-1 matrix, then $\TT\LL \AA$ stands for the result of the
following process: first multiply
 the matrices $\TT$ and $\LL$ in the standard way, obtaining a matrix $\TT'$ of new homogenous linear forms, and then multiply $\TT'$
by the vector $\AA$, as explained above.
We also denote by $\widehat{ \AA}$ the corresponding vector of non-commutative
\emph{polynomials} computed by the coordinates in $\AA$.
\medskip

With these notations in hand, we now 
construct the following ABP-variant of the identity witnesses:

\begin{lemma}[existence of ABP-based identity witnesses]\label{lem:Raz-ABP}
If the ABP $A(\vin,\vsink)$ computes the  identically zero non-commutative polynomial, then the following hold:

\begin{enumerate}
\item 
\label{it:some_existence_lemma} There exist $\,l$  matrices $\LL_i$ with 0-1 entries, for $i=0\. l-1$, each  of dimension $m_i\times m_i$, where $m_i={\rm
poly}(|A|)$, 
such that  $\LL_0=1$ and 
  $$\LL_{l-i}  \AA_i = \overline 0,~~\text{for all $i=0\.l-1$}$$
(where the equality here is only \emph{semantic}, i.e., the left hand side computes
a vector of zero non-commutative polynomials).

\item 
There exist $l-1$ matrices $\TT_i$, for $i=1\. l-1$, of dimension $m_{i-1}\times m_{i}$ and whose entries are homogenous linear forms in the  $\overline x$ variables with 0-1 coefficients, such that
\begin{gather}\label{eq:18}
\LL_{l-i}\AA_i = 
   \TT_{l-i+1} \LL_{l-i+1}\AA_{i-1},~\hbox{~for $i=2,\ldots,d$, ~and}
\\ \label{eq:final_level}
\LL_    {0}\vectorABPLayerFinal = A(\vin,\vsink) \,,
\end{gather}
and where the ABPs in these  two equalities are constructed in the way described above (these two equalities above are \emph{syntactic}, i.e., in each of the equations the two sides are syntactically identical as ABPs).
\end{enumerate}
\end{lemma}
\begin{proof}
Recall that $m_i$ is the number of nodes in the $i$th layer. Since we assumed $\Lambda_{0}=1$, and since $\vectorABPLayerFinal
= A(\vin,\vsink)$, we conclude equation \eqref{eq:final_level} in the lemma.

We now construct by induction on $j$ the matrix  $\LL_j$, for $j=0\.l-2$, such that part \ref{it:some_existence_lemma} in the lemma holds:
\begin{equation}\label{eq:XXX}
\LL_j\AA_{l-j}=0,
\end{equation}
as well as  \eqref{eq:18}, that is,
\begin{equation}
\label{eq:YYY}
\LL_{j}\AA_{l-j}=\TT_{j+1}\LL_{j+1}\AA_{l-j-1}\,.
\end{equation}

\Base  $\LL_0=1$ by assumption.

\induction Assume that for $0\le h<l-1,~ \LL_0\.\LL_{h-1}$ and $\TT_0\.\TT_{h}$ were already constructed, and that the equality
\eqref{eq:XXX} holds for every $j=0\. h$, and equality  \eqref{eq:YYY} holds for every $j=0\. h-1$.
We will now construct $\LL_{h+1}$ such that \eqref{eq:XXX} holds for $j=h+1$:
$$
\LL_{h+1}\AA_{l-h}\,,
$$
and  $\TT_{h+1}$ such that \eqref{eq:YYY} holds with $j=h$:
$$
\LL_{h}\AA_{l-h}=\TT_{h+1}\LL_{h+1}\AA_{l-h-1}\,.
$$
Let $M_{h,h+1}$ be the adjacency matrix of dimension $m_{h}\times m_{h+1}$ of the two consecutive layers $h$ and $h+1$ in $A,$  where for each entry $(p,q)$, for $p\in[m_{h}],q\in[m_{h+1}]$,
\begin{gather*}
M_{h,h+1}(p,q)=A(v_{h,p},v_{h+1,q})=\sum_{k=1}^n c_kx_k,~~~~~\hbox{where }c_k\in\set{0,1}.
\end{gather*}
The matrix $M_{h,h+1}$ can be written as $\sum_{k=1}^n x_k M_{h,h+1}^{k} $ (the superscript $k$ is used here as an \emph{index} only, and not as a matrix power), for some 0-1 matrices $M_{h,h+1}^{k}$.
By the definition of an ABP
\begin{align*}
  \AA_{l-h}&=M_{h,h+1}\AA_{l-h-1}\\
  &=\sum_{k=1}^n x_k M_{h,h+1}^{k}\AA_{l-h-1}\,.
\end{align*}
Moreover,  if $\LL_{h}\AA_{l-h}=\overline 0 $,   then
$$\LL_{h}\sum_{k=1}^n x_k  M_{h,h+1}^{k}\AA_{l-h-1}=0,$$
and therefore, by  the non-commutativity of product we have
\begin{equation}
\label{eq:E}
\LL_{h} M_{h,h+1}^{k}\AA_{l-h-1}=\overline 0, ~~~\text{ for }k=1\. n\,.
\end{equation}

Now, consider the basis of the span of all row vectors in  all the matrices $\LL_h M_{h,h+1}^{k}$,
for $k=1\. n$. The number of vectors in this  basis is at most the number of columns in (each of the) $M_{h,h+1}^k$ matrices, that is, at most  $m_{h+1}$  (which equals the number of nodes in the $h+1$ layer). Define
the matrix $\LL_{h+1}$ to be the $m_{h+1}\times m_{h+1}$ matrix whose rows are the vectors in this basis (if the basis consists of less than $m_{h+1}$ vectors we can simply put zero rows to reach $m_{h+1} $ rows).
By \eqref{eq:E} we know that every row vector in $\LL_{h} M_{h,h+1}^{k}$ is \textit{orthogonal} to $\AA_{l-h-1}$
(i.e.,~their inner product is zero).  Thus, any vector in the basis of the rows of $\LL_{h} M_{h,h+1}^{k}$,
for $k=1,\ldots,n$,
is also orthogonal to $\AA_{l-h-1}$, and so  $$\LL_{h+1}\AA_{l-h-1}=\overline 0\,.$$

By the properties of a basis of a linear space, there must exist  matrices $\TT_{h+1}^{k}$, such that
$$\LL_{h} M_{h,h+1}^{k}=\TT_{h+1}^{k}\LL_{h+1},\text{~~~~for all~} k=1\. n\,. $$
Then, define $\TT_{h+1} := \sum_{k=1}^n \TT_{h+1}^{k} x_k$.
Thus,
$$\LL_{h}\AA_{l-h}=\TT_{h+1}\LL_{h+1}\AA_{l-h-1}\,.$$
\end{proof}

\newcommand{\nonHomoABP}{\ensuremath{A}}
\newcommand{\gl}{\g_{left}}
\newcommand{\gr}{\g_{right}}
\newcommand{\vg}{{\g_v}}

\subsubsection{Implicitly working with ABPs in Frege system}
\label{sec:sub-sec-on-implicit-ABP}

Here we use Lemma \ref{lem:Raz-ABP} to conclude the existence of (formula-based)
identity witnesses (as required in Theorem \ref{existLambda}) and by that conclude
the proof of Theorem \ref{thm:Fbool-in-quasipolynomial-Frege}.


Recall the notion of an induced part of a
formula (Definition \ref{def:part-of-formula}): for a subformula $F'$ of $F$ and gates  $g_1\.g_k$  in $F'$ and 0-1 constants $c_1\.c_k$,  $F'[c_1/g_1,\.,c_k/g_k]$ is called  \emph{an induced part of $F$}.  
Notice that  it is unclear how to (usefully) represent an ABP directly in a Frege system, because apparently ABP is a stronger model than formulas (and each Frege proof-line is written as a formula). Thus, we cannot directly use the same formulation as \cite{RS04}.
This is the reason that we work with induced parts of formulas: let  $A$ be the ABP that corresponds to the non-commutative formula
$F$, then for every node $v$ in $A$ there will be a corresponding induced part of $F$ that computes the same polynomial computed by  the sub-ABP rooted in $v$ and whose sink is the sink of  $A$. For this purpose we introduce the following notation and definition.  \smallskip

For two vertices $v',v''$ in the ABP $A$, we denote by $A(v',v'')$ the polynomial computed by the ABP with the source $v'$ and the sink $v''$ and all the paths leading from $v'$ to $v''$.
Informally, a \emph{$v$-part of a formula $F$} is simply a substitution instance of $F$ that computes
the same polynomial as $\nha(v,\vsink)$.
Formally we have:

\begin{definition}[$v$-part of formula  $F$]
Let $F$ be a homogenous formula and $A$ be the corresponding ABP of $F$ constructed according to the methods described in \cite{RS04} (see also below), in which the source is $\vin$ and the sink is $\vsink$. For any node $v$ in $A$, if there exists an induced part of the formula $F$ computing the same polynomial as $A(v,\vsink)$, then we call this part a \emph{$v$-part of the formula $F$}. (Note that for a node $v$ there might be
more than one $v$-part.) 
\end{definition}

Let $F$ be a non-commutative homogenous formula and let $A$ be the corresponding ABP of $F$. In (the
proof of) Lemma \ref{existence-v-part} below we construct a mapping between the nodes $v$ in $A$ to $v$-parts
of $F$, denoted $\nodeF_v$ such that  $\nodeF_v$ computes the non-commutative homogenous polynomial   computed by $A(v,\vsink)$.
This will enable us to refer (implicitly) to $A(v,\vsink)$  by an induced part $\nodeF_v$ of $F$, for any node $v$ in $A$ (though a $v$-part is not unique, our mapping will obviously associate a \textit{unique} $v$-part to every node $v$ in $A$). 

Furthermore,  for every node $v$ in the ABP $A$ the following holds (where, for the sake of simplicity, the arithmetic formulas computing the linear forms computed by $A(v,u)$ are denoted also  by  $A(v,u)$):
\begin{equation}\label{eq:big_sum}
\nodeF_{v}\synEqual 
\sum_{u: {\text{ $u$ has an incoming }}\atop \text{edge from  $v$}}
A(v,u)\times\nodeF_{u}.
\end{equation}
where, the big sum denotes a  balanced binary tree of plus gates and the $A(v,u)\times \nodeF_u$'s at the leaves. 

It will be convenient to assume that the sink \vsink\ of an ABP is mapped to the empty formula and that if $G$ is the empty formula, then $H\times G \synEqual H$, where $H$ stands for some nonempty formula. 


\begin{lemma}\label{existence-v-part}
 For a non-commutative syntactic homogenous formula $F$  without constants, let  $A$ be the ABP transformed from $F$ by the methods in \cite{RS04} (equivalently, in \cite{Nis91};
we repeat this construction in the proof  below), in which the source is $\vin$ and the sink is $\vsink$. For every node $v$ in $A$ the non-commutative polynomial computed by $A(v,\vsink)$ can be computed by some non-commutative formula, denoted $\nodeF_v$, which is
a v-part of $F$. Furthermore, for every node $v$ in the ABP,  \eqref{eq:big_sum} holds.
\end{lemma}

\begin{proof}
We construct an  ABP \nonHomoABP, such that each node $v$ in \nonHomoABP\ is 
mapped to an induced part of $F$.


\para{Constructing the ABP \nonHomoABP.}
Given the non-commutative syntactic homogenous formula $F$ over $GF(2)$ that does not contain constants, we construct the corresponding ABP \nonHomoABP\ by induction on the size of $F$. By syntactic homogeneity \textit{we get a standard (layered) ABP} (this differs from  \cite{RS04} who  did not start from  a homogenous formula and
so the resulted ABP was (initially) non-layered). Throughout the construction of \nha\ we maintain a mapping
$$
g: \rm{nodes}(\nonHomoABP)\to\rm{nodes}(F)
$$ from nodes in \nonHomoABP\
to their ``corresponding'' nodes in $F$ (this will help us define  the mapping $\nodeF_v$).
The reader can also consult the illustrated example in the sequel.

\Base If $F$ is a variable $x_i$, then \nonHomoABP\ is a single edge $(\vin,\vsink)$ labeled with $x_i$
and $g(\vin):=F$ (i.e., the single node in $F$) and  $g(\vsink):=\emptyset$ (i.e., ``the empty node'').

\induction
 
\case 1 $F=G+H$. Then \nonHomoABP\ is defined with the root \vin\ being the joint of the two roots of
the two ABPs constructed already for $G,H$ (while keeping their outgoing edges).
We then also join the two sinks of the ABPs for $G,H$ (while keeping their incoming edges) into a single sink denoted \vsink.

The function $g$ is defined as the union of the two original functions $g$'s for the ABPs for $G$ and $H$, where the
new nodes \vin\ and \vsink\ are mapped by $g$ to the root of $F$ and the empty formula $\emptyset$, respectively
(note that
the domains of both these $g$'s are \emph{disjoint}---except for the two sources and two sinks).
\medskip

\case 2 $F=G\times H$. Assume that $ A_G, A_H$ are the two ABPs already constructed
for $G,H$, respectively. Then \nonHomoABP\ is defined as $ A_G$ with the sink of $  A_G$ replaced by $ A_H$. The function $g$ is defined as the union of the two $g$
functions for  $ A_G,  A_H$ (where the root of $ A_H$ is mapped by $g$ to the root of the formula
$H$).

\para{Construction of $F^\bullet_v$.}

Let $F$ be a syntactic homogenous non-commutative formula without constants, \nha\ its corresponding ABP, and $g:{\rm nodes}(\nha)\to{\rm nodes}(F)$ the  function, all whose construction is described above. We construct the mapping  $F^\bullet_v$ for nodes $v$ in \nonHomoABP\ 
as follows. Throughout the construction we maintain the following conditions:
%
\begin{itemize}
\item[(i)]for every node $v$ in \nonHomoABP, the formula $F^{\bullet}_v$ computes the same non-commutative
polynomial as
$\nonHomoABP\left(g(v),\vsink\right)$;
\item[(ii)] for every node $v$ in $\nha$, equation \eqref{eq:big_sum} above holds. 
\end{itemize}
  
\newcommand{\fbv}{\ensuremath{F^{\bullet}_v}}

Let $F$ be a non-commutative syntactic homogenous formula $F$, and  $t$ a node in $F$. Denote by $r$
the root of $F$ (in particular, if $F$ is a variable then the root is the variable). We define the function $D(F,t)$ by induction on the structure
of $F$ as follows:
$$D(F,r):=F,$$ 
and for $t\neq r$ we define (for $A, B$ two non-commutative syntactic homogenous formulas without constants): 
\[
D(A+B, t):=
\begin{cases}
D(A,t)+0\,, & \hbox{ if $t\in {\rm nodes}(A)$;}  \\
0+D(B,t)\,, &  \hbox{ if $t\in {\rm nodes}(B)$,} \\
\end{cases}
\]
and
\[
D(A\times B, t) :=
\begin{cases}
D(A,t)\times B\,, & \hbox{ if $t\in {\rm nodes}(A)$;}  \\
1\times D(B,t)\,, &  \hbox{ if $t\in {\rm nodes}(B)$} \\
\end{cases}
\]
(note the asymmetry in defining $D(A\times B)$, which corresponds to the way a non-commutative formula
is translates into an ABP, with $A$ computed ``above'' $B$).

Finally, for every node $v$ \emph{in the ABP $A$}, we define  
$$
\fbv := D(F,g(v)).
$$

\para{Example.} Figure \ref{fig:XXX} illustrates a non-commutative syntactic homogenous and constant-free formula $F$, and its corresponding ABP $A$, together with the map $g:{\rm nodes}(A)\to {\rm nodes}(F)$. Figure \ref{fig:YYY} shows the formula $\fbv$ (where $v$ is the node in $A$ from Figure \ref{fig:XXX}). Note that indeed, by definition, $D(F,g(v))=D(F,t)=D(F_s,t)\times
F_q = (D(F_p,t)+0)\times F_q = ((1\times D(F_t,t))+0)\times F_q = ((1\times x_2)+0)\times F_q $.
 
\begin{figure}[here]
\begin{center}
\input{drawing.tex}
\caption{}
\label{fig:XXX}
\end{center}
\end{figure}

\begin{figure}[here]
\begin{center}
\input{substitution-instance.tex}
\caption{}
\label{fig:YYY}
\end{center}
\end{figure}

\begin{claim*}\label{cla:}
Conditions (i) and (ii) above hold.   
\end{claim*}

\begin{proofclaim}
Condition (i) holds by inspection of the definition of $D$, the construction of the ABP \nha\ from $F$
and the function $g$ giving the ``origin'' in $F$ of each node in \nha.
\smallskip 

Condition (ii), i.e., equation \eqref{eq:big_sum} for all $v$ in $\nha$ holds by condition
(i) and the definition of $D$. Note that condition (i) already shows 
that 
$
\nodeF_{v} \sequiv \sum\limits_{u: {\text{ $u$ has an incoming }}\atop \text{edge from  $v$}}
A(v,u)\times\nodeF_{u}\,
$ 
is indeed a tautology (considered over $GF(2)$). The fact that the right hand side of this tautology can be derived with a short (polynomial-size) Frege proof from the left hand side can be demonstrated by using basic structural derivation rules of Frege (e.g.,
associativity and distributivity) and simple logical equivalences (e.g., $1\oplus G\sequiv \neg G$). We omit the details.
\end{proofclaim}

\medskip 
This concludes the proof of Lemma \ref{existence-v-part}.
\end{proof} 

We are now ready to conclude the proof of the identity witnessing theorem (Theorem \ref{existLambda}):

\newcommand{\Ai}{\ensuremath{\overline A _i}}

\begin{proof}[Proof of Theorem \ref{existLambda}]
Recall the ABP-based identity witnesses we showed existed in 
Lemma \ref{lem:Raz-ABP}. Our goal is to show that there are (formula-based) identity witnesses (as defined in  the proof of Theorem \ref{thm:homo-formula}). 

Using the correspondence given in  Lemma \ref{existence-v-part},
we can replace each ABP $A(v,\vsink)$ occurring in some $\Ai$ (for some node $v$ in $A$ and some $i=0,\ldots,l$) by a corresponding $v$-part $\nodeF_v$. Denote with  $\homoji$  the result of this replacement. Thus, $\homoji$ contains $v$-parts  $\nodeF_v$ of $F$, each computing a homogenous polynomial of degree $i$. 
With this replacement we get (we assume that $l=d$):
\begin{enumerate}
\item 
 There exist $d$ matrices $\LL_i$, for $i=0\. d-1$, with 0-1 entries and dimension $m_i\times m_i$, where $m_i={\rm poly}(|A|)$, such that  $\LL_0=1$, $\overline F_d=F$
and   
$$
\LL_{d-i}  \FF_i = \overline 0,~~\text{for all $i=0\.d-1$}.
$$

\item 
There exist $d-1$ matrices $\TT_i$, for $i=1\. d-1$, of dimension $m_{i-1}\times m_{i}$ and whose entries are homogenous linear forms in the  $\overline x$ variables with 0-1 coefficients, such that 
\begin{gather}\label{eq:18:later}
\LL_{d-i} \FF_i = 
   \TT_{d-i+1} \LL_{d-i+1}\FF_{i-1},~\hbox{~for $i=2,\ldots,d$, ~and}
\\ \label{eq:final_level:later}
\LL_{0}\FF_d = A(\vin,\vsink) \,.
\end{gather}
\end{enumerate}

Recall that $\Lambda \overline F$, for $\Lambda$ a matrix and $\overline F$ a vector of formulas, is a vector of formulas, where each formula is written as a balanced
(partial) sum of the formulas in $\overline F$ (see the Identity Witnesses' definition
in the proof of Theorem \ref{thm:homo-formula}).


%
%

Our goal now is to show 
\begin{equation}\label{eq:almost_end}
 \synEqual \LL_{d-i}\homoji  \sequiv \TT_{d-i+1}\LL_{d-i+1}\FF_{i-1}\,,~~~~~~i=2\.d\,,
\end{equation}
and 
\begin{equation}\label{eq:don't_need}
  \synEqual F \sequiv  \LL_0\homojiFinal \hbox{ (meaning that $\synEqual F \sequiv F_d$,
since $\LL_0=1$).}
\end{equation}
Note that \eqref{eq:don't_need} holds trivially since $F_d$ and $F$ are identical. For \eqref{eq:almost_end},
by Lemma \ref{existence-v-part}, we know that 
\begin{equation}\label{eq:is_it_last?}
\nodeF_{v}\synEqual 
\sum_{u: {\text{ $u$ has an incoming }}\atop \text{edge from  $v$}}
A(v,u)\times\nodeF_{u}.
\end{equation}
Consider $\LL_{d-i}\homoji$. Each of its entries is a (partial) sum of the formulas
in $\homoji$ written as a balanced sum. By \eqref{eq:is_it_last?} we can write each entry in $\homoji$ as a (balanced) sum in which each summand is some linear form $A(v,u)$ times an entry in $\FF_{i-1}$.  We can thus write $\LL_{d-i}\homoji$ as $\TT'\FF_{i-1}$
for some matrix $\TT'$ with linear forms in each of its entries. Since the identity stated in \eqref{eq:almost_end} is a \emph{true} identity, we thus get 
$$
\TT'\FF_{i-1} \sequiv \TT_{d-i+1}\LL_{d-i+1}\FF_{i-1}\,.
$$
But such an identity is provable in Frege with a polynomial-size proof, because we
only need to prove an identity between $\langle \mathbf t_j',\FF_{i-1}\rangle $ and $\langle \mathbf t_j, \FF_{i-1}\rangle $, for each of the $j$th rows $\mathbf t'_j$ and $\mathbf t_j$ of the matrices $\TT'$ and $ \TT_{d-i+1}\LL_{d-i+1}$, respectively
(note that each entry of these rows is written as a \textit{linear} form).
\end{proof}

\subsection{Conclusions}
The propositional-calculus has a ubiquitous presence in logic and computer science
at large. Within complexity theory and  propositional proof complexity in
particular it has a prominent role, and considered a strong proof system whose structure and complexity is  poorly understood. In that respect, we believe our characterization of Frege proofs and the  propositional-calculus as non-commutative
polynomials whose non-commutative formula size corresponds
(up to a quasi-polynomial increase) to the size of  Frege proofs,
should be considered a valuable contribution.

In the framework of algebraic propositional proof systems (and especially
the IPS framework and its precursors by Pitassi \cite{Pit97,Pit98}) our 
characterization is \emph{almost precise}, as we showed an almost tight two-sided simulation of Frege and non-commutative IPS. Although we left it  open whether the simulation
of non-commutative IPS by Frege can be improved from quasi-polynomial down to polynomial
size,
there is nothing to suggest at the moment this cannot be achieved.       

Non-commutative formulas constitute a weak model of computation that is quite well understood. Since, as mentioned above,  the Frege system is considered a strong proof system, and in fact it is not entirely out of question  that Frege---or at least its extension, Extended Frege---is polynomially bounded (i.e., admits polynomial-size proofs  for every tautology),
on the face of it, our results are  surprising.

Overall, we believe that this correspondence between non-commutative formulas and proofs, gives a renewed hope for progress on the fundamental lower bounds questions in proof complexity
in so far that  it reduces the problem of proving lower bounds on Frege proofs to the problem of establishing non-commutative formula lower bounds (which are already
known for many polynomials). Since non-commutative lower bounds are already known and since proving the sort of matrix rank lower bounds that are required to establish  non-commutative formula lower bounds for  the permanent and determinant are fairly simple (\cite{Nis91}),
the current work provides a quite compelling evidence that Frege lower bounds
might indeed not be very far away.

One possible route for Frege lower bounds is to reduce
directly Frege lower bounds to the problem of lower bounding the non-commutative formula
size of polynomials that are \emph{already known to be hard} for the class of non-commutative formulas. We believe that this route is certainly a plausible one. A somewhat less direct approach is to show that some tautologies require non-commutative
IPS refutations whose associated partial-derivative matrices  (in the sense of Nisan \cite{Nis91}) have high  rank---here, the task would be to lower bound non-commutative
polynomials that are given only ``semi-explicitly'' (that is, they are given in terms of the properties of the non-commutative IPS (Definition \ref{def:intro:non-commutative-IPS}));
in other words, one has to establish lower bounds on a \emph{family }of polynomials (for
each fixed number of variables $n$).

Furthermore, ideas and lower bounds techniques connecting non-commutative computation,
algebras with polynomial-identities (PI-algebras) and proof complexity as studied in \cite{Hru11,LT13} might provide
further tools for obtaining non-commutative IPS lower bounds.


\medskip 

Apart from the fundamental lower bound questions, the new characterization of Frege proofs sheds new light on the correspondence between \textit{circuits and proofs} within proof complex    ity: in the framework of the ideal proof system, a Frege proof can be seen from the computational perspective as a non-commutative formula. This gives a different, and in
some sense simpler,  correspondence between proofs and computations than
the traditional one (in which Frege corresponds to \NCOne\ (cf.~\cite{CN10})).    

We have also tighten the important results of Grochow and Pitassi \cite{GP14}. Namely, by showing that already the non-commutative version of the  IPS is sufficient to simulate Frege, as well as by showing \textit{unconditional }efficient simulation of the non-commutative IPS by Frege. 

Finally, while proving that Frege quasi-polynomially simulates the non-commutative IPS, we demonstrated new simulations of algebraic complexity constructions within  proof complexity; these include the homogenization for formulas due to  Raz \cite{Raz13-tensor} and the PIT algorithm for non-commutative formulas due to  Raz and Shpilka \cite{RS04}. These proof complexity  simulations add to the known previous such simulations shown in Hrube\v s and the second author \cite{HT12}, and are of independent interest in the area
of Bounded Arithmetic and feasible mathematics.

\section*{Acknowledgments}
We are thankful to Joshua Grochow for very helpful comments.

\bibliographystyle{alpha}
\bibliography{PrfCmplx-Bakoma}

\end{document}